\documentclass[twocolumn,superscriptaddress]{revtex4}
\usepackage{graphicx}
\usepackage{epstopdf}
\usepackage{dcolumn}
\usepackage{bm}
\usepackage{amsthm}
\usepackage{amsfonts}
\usepackage{color} 
\usepackage[usenames,dvipsnames]{xcolor}
\usepackage{amscd}
\usepackage{amsmath}
\usepackage{enumerate}
\usepackage{bbm}
\usepackage{hyperref}
\usepackage{float}
\usepackage{mathtools}
\usepackage{amsthm}
\usepackage{subfigure}

\DeclareMathOperator{\tr}{Tr}
\newcommand{\ie}{\textit{i.e.}}

\newcommand{\be}{\begin{equation}}
\newcommand{\ee}{\end{equation}}

\newcommand{\sket}[1]{{\ensuremath{\lvert#1\rangle}}}
\newcommand{\lket}[1]{{\ensuremath{\left\lvert#1\right\rangle}}}
\newcommand{\ket}[1]{\if@display\lket{#1}\else\sket{#1}\fi}
\newcommand{\ketMode}[2]{{\ensuremath{\lvert #1\rangle _{#2}}}}
\newcommand{\sbra}[1]{{\ensuremath{\langle#1\rvert}}}
\newcommand{\lbra}[1]{{\ensuremath{\left\langle#1\right\rvert}}}
\newcommand{\bra}[1]{\if@display\lbra{#1}\else\sbra{#1}\fi}
\newcommand{\braMode}[2]{{\ensuremath{\langle #1\rvert _{#2}}}}
\newcommand{\sbraket}[2]{{\ensuremath{\langle#1\rvert#2\rangle}}}
\newcommand{\lbraket}[2]{{\ensuremath{\left\langle#1\!\left\rvert\vphantom{#1}#2\right.\!\right\rangle}}}
\newcommand{\braket}[2]{\if@display\lbraket{#1}{#2}\else\sbraket{#1}{#2}\fi}

\newcommand{\sketbra}[2]{{\ensuremath{\lvert #1\rangle\!\langle #2\rvert}}}
\newcommand{\lketbra}[2]{{\ensuremath{\left\lvert #1\right\rangle\!\!\left\langle #2\right\rvert}}}
\newcommand{\ketbra}[2]{\if@display\lketbra{#1}{#2}\else\sketbra{#1}{#2}\fi}
\newcommand{\ketbraMode}[3]{{\ensuremath{\lvert #1\rangle _{#3}\langle #2\rvert}}}

\newcommand{\proj}[1]{\ketbra{#1}{#1}}

\theoremstyle{nonumberplain} 
\newtheorem*{claim*}{Claim}

\newtheorem*{cor*}{Corollary}

\theoremstyle{definition}

\begin{document}

\title{Beating the repeaterless bound with adaptive measurement-device-independent quantum key distribution}
\begin{abstract}
Surpassing the repeaterless bound is a crucial task on the way towards realizing long-distance quantum key distribution. In this paper, we focus on the protocol proposed by Azuma \textit{et} \textit{al.} in [Nature Communications 6, 10171 (2015)], which can beat this bound with idealized devices. We investigate the robustness of this protocol against imperfections in realistic setups, particularly the multiple-photon pair components emitted by practical entanglement sources. In doing so, we derive necessary conditions on the photon-number statistics of the sources in order to beat the repeaterless bound. We show, for instance, that parametric down-conversion sources do not satisfy the required conditions and thus cannot be used to outperform this bound.
\end{abstract}
\author{R\'obert Tr\'enyi}
\affiliation{Escuela de Ingenier\'ia de Telecomunicaci\'on, Dept. of Signal Theory and Communications, University of Vigo, E-36310 Vigo, Spain}
\author{Koji Azuma}
\affiliation{NTT Basic Research Laboratories, NTT Corporation, 3-1 Morinosato Wakamiya, Atsugi, Kanagawa 243-0198, Japan}
\affiliation{NTT Research Center for Theoretical Quantum Physics, NTT Corporation, 3-1 Morinosato-Wakamiya, Atsugi, Kanagawa 243-0198, Japan}
\author{Marcos Curty}
\affiliation{Escuela de Ingenier\'ia de Telecomunicaci\'on, Dept. of Signal Theory and Communications, University of Vigo, E-36310 Vigo, Spain}
\maketitle

\section{Introduction}

Quantum key distribution (QKD) protocols can provide two distant parties (Alice and Bob) with information-theoretically secure secret keys~\cite{GisinReview, NorbertReview, MarcosReview}. However, in point-to-point QKD via pure-loss bosonic channels, the achievable secret key rate is fundamentally limited by the so-called repeaterless bound~\cite{TGW, PLOB, AnotherPaperOnBounds}. In the limit of high channel loss (\ie , long distances) the repeaterless bound is proportional to the transmittance of the channel connecting Alice and Bob, denoted by $\eta$. This means that the secret key rate of any point-to-point QKD protocol scales at most with $\eta$. As in the case of optical fibers $\eta=e^{-L/L_{\rm{att}}}$, where $L$ is the distance between the parties and $L_{\rm{att}}$ is the attenuation length of the fiber, this poses a stringent limitation on the achievable secret key rate. Therefore, surpassing the repeaterless bound is an essential step towards efficient long-distance quantum communication.

A simple idea to outperform the repeaterless bound is to introduce intermediate nodes, dividing the channel into many smaller segments so that the probability of losing a signal stays relatively small on each segment. This naturally leads to the concept of quantum repeaters~\cite{BriegelRepeater, DuanRepeater, HybridRepeater, KokRepeater, JiangRepeater, SangouardRepeater,
KojiRepeater, MunroRepeater, KojiRepeater2}, which are typically based on entanglement swapping and distillation. However, to truly benefit from a quantum repeater, one needs many nodes and demanding technological resources, which makes the experimental realization very challenging with current technology~\cite{RepeaterDifficult}.

Another, recently proposed idea is twin-field QKD (TF-QKD)~\cite{LucamariniTfQkd}, which is based on single-photon interference and includes one intermediate node performing such conceptually simple interferometric measurement. Indeed, this offers a square root improvement in the scaling of the secret key rate. This proposal has triggered a lot of attention in the field and various simple security proofs and improved versions of the original protocol have been very recently introduced~\cite{MarcosTf, CuiTF, JieTF, kiyoshiTF, MaTF, WangTF, YinTF}. Proof-of-principle experiments to show the feasibility of some of the suggested TF-type QKD protocols have already been demonstrated experimentally~\cite{ProofOfPrincipleTF, ToshibaExpTF, PanTF, WangTFExp}, which may suggest the viability of this approach to achieve intercity QKD with current technology. 

The same square root improved scaling can be achieved if one extends the original measurement-device-independent QKD (MDI-QKD)~\cite{mdi} protocol, based on two-photon interference, with some feedback mechanism to make sure that the Bell state measurement (BSM) is performed between signals that actually survived the channel loss. One way is to make use of quantum memories~\cite{luong, abruzzo, panayi}, but the required memory parameters are still challenging for current technology~\cite{ posterQCrypt}.
To avoid the need for quantum memories while having the same square root improved scaling, Azuma \textit{et} \textit{al}. proposed the idea of a fully optical, adaptive MDI-QKD~\cite{koji} (AMDI-QKD) protocol, using standard optical teleportation for performing a quantum non-demolition measurement (QND)~\cite{qnd} to confirm the arrival of the single-photon signals at the middle node. While the required technology to implement the AMDI-QKD protocol is far off our current experimental capabilities, this protocol could offer higher secret key rates than the TF-QKD protocol, because the former is based on two-photon interference at the middle node while the latter is based on single-photon interference~\cite{HadamardOperation}.

The original AMDI-QKD scheme~\cite{koji} assumes highly idealized devices, like, for instance, perfect entanglement sources, which are capable of generating a perfect EPR pair on demand for the teleportation in the QND measurement, and perfect single-photon sources. In this paper, we investigate the robustness of this protocol against source imperfections, like for example a non-vanishing probability of emitting multiple-photon signals and thus introducing extra noise into the system. By performing a full-mode analysis of a realistic setup we derive a necessary condition on the photon-number statistics of the sources for overcoming the repeaterless bound in~\cite{PLOB}.

In doing so, we also show, for example, that with parametric down-conversion (PDC) sources, the AMDI-QKD protocol has a scaling of at most $\eta$, therefore unable to  beat the repeaterless bound. This is due to the fact that PDC sources have a too large probability of emitting multiple-photon pairs compared to the probability of emitting single-photon pairs.

We note that a similar behavior has also been observed in the context of the ensemble-based quantum memory assisted MDI-QKD~\cite{mohsen} protocol. Ensemble-based quantum memories have many favorable properties, but they inherently suffer from a non-negligible probability of emitting multiple-photons (similar to that of the PDC sources) causing the advantageous scaling offered by a traditional memory assisted system~\cite{luong, abruzzo, panayi} to vanish. In this regard, we remark that our result is stronger than that introduced in~\cite{mohsen} in the sense that it applies even with photon-number resolving (PNR) detectors, while~\cite{mohsen} assumes threshold detectors.

The paper is organised as follows. In Sec.~\ref{Sec:ProtocolDescription}, the investigated protocol is introduced and its secret key rate formula is presented. Sec.~\ref{Sec:ToolBox} describes mathematically the physical devices used for the implementation of the protocol. Next, in Sec.~\ref{Sec:NecessaryConditions} we present the main results of the paper. Here, we obtain necessary conditions on the applied entanglement sources for overcoming the repeaterless bound~\cite{PLOB} with the AMDI-QKD protocol. As a corollary, we prove that the protocol is not capable of beating the repeaterless bound~\cite{PLOB} using PDC sources. Lastly, Sec.~\ref{Sec:Conclusion} contains the conclusions of the paper. The paper also includes two appendices for providing the details of the calculations.

\section{The AMDI-QKD protocol}\label{Sec:ProtocolDescription}

\subsection{Protocol steps}\label{Sec:steps}
The schematic layout of the AMDI-QKD protocol~\cite{koji} can be seen in Fig.~\ref{layout}. The protocol runs as follows:

\begin{enumerate}
  \item Each of Alice and Bob generates $m$ signals with their on-demand entanglement sources $S_{\rm{AC}}$ and $S_{\rm{BC}}$, respectively. One mode of each signal is sent to Charlie's QND measurement simultaneously via the quantum channel, using a multiplexing technique (e.g. wavelength based). The other mode is kept by Alice and Bob and they measure it in the $Z$ or $X$ basis, which they choose with probabilities $p_{\rm{Z}}$ and $p_{\rm{X}}=1-p_{\rm{Z}}$, respectively. 
  \item Charlie applies QND measurements to the incoming pulses to confirm the arrival of the signals coming from Alice and Bob. 
  \item Charlie pairs the successfully arriving signals via optical switches and performs BSMs between signals coming from the different parties. To be more precise, if there are, say $n_{\rm{A}}$ ($n_{\rm{B}}$) successfully arriving signals from Alice (Bob), then Charlie performs $\min(n_{\rm{A}},n_{\rm{B}})$ BSMs. 
  \item Charlie announces to Alice and Bob which BSMs were successful, together with the measurement result obtained. Here the successful BSM is assumed to distinguish the Bell states $\ket{\psi^-}=1/\sqrt 2(\ket{HV}-\ket{VH})$ and $\ket{\phi^-}=1/\sqrt 2(\ket{HH}-\ket{VV})$ from the others, as we consider (see Fig.~\ref{BsmLayout} in Sec.~\ref{Sec:ToolBox}) a standard linear-optics implementation of the BSMs. Here $H$ ($V$) represents a horizontally (vertically) polarized single-photon state.
  \item For key generation, Alice and Bob post-select the events by communicating over an authenticated classical channel where they used the $Z$-basis to measure their modes in step 1 (and in which they both had a successful photon detection) and also the corresponding BSM was successful. To make sure that their key bits are identical, they apply a bit flip procedure~\cite{mdi}. To be precise, Alice or Bob flips her or his bits, except for the cases when they chose the $Z$ basis and Charlie's BSM outcome was the state $\ket{\phi^-}=1/\sqrt 2(\ket{HH}-\ket{VV})$.  
\end{enumerate}
 
\begin{figure}[H]
\centering  
  \scalebox{0.26}{\includegraphics{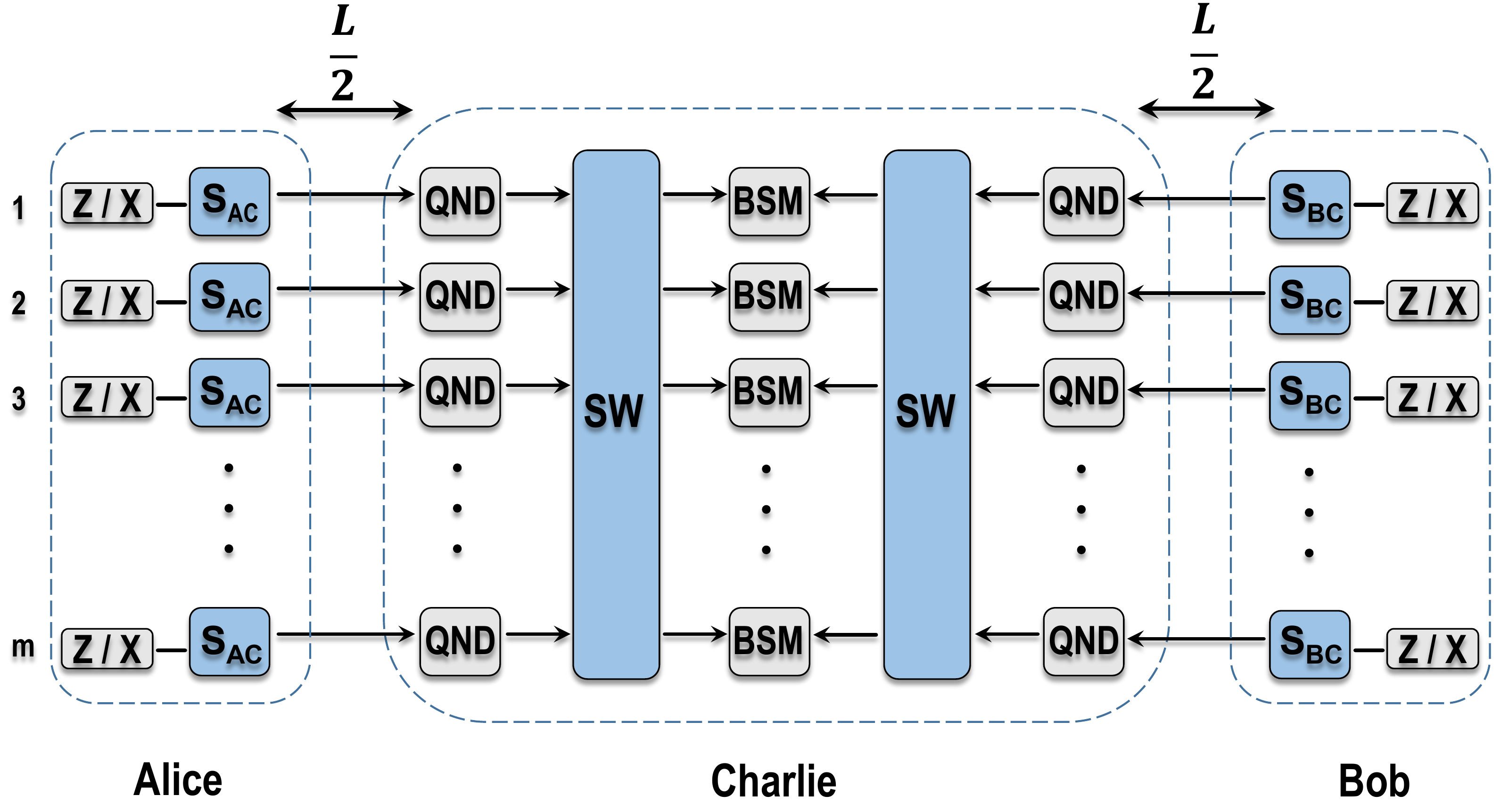}}
  \caption{Schematic layout of the AMDI-QKD protocol using the entanglement sources $S_{\rm{AC}}$ and $S_{\rm{BC}}$. The distance between Alice and Bob is denoted by $L$ and we assume that the untrusted middle node Charlie is located halfway between them. The parameter $m$ is the number of pulses that each party sends to Charlie. A QND measurement is performed by Charlie to confirm the arrival of the signals emitted by the $S_{\rm{AC}}$ and $S_{\rm{BC}}$ sources. Also, Alice and Bob measure their modes of the sources in the $Z$ or $X$-basis, which they select probabilistically, indicated in the figure by the symbol (Z/X). Conditioned on the success of the QND measurements, the surviving signals are paired via optical switches (SW) and Bell state measurements (BSM) are performed between all the pairs by Charlie.}\label{layout}
\end{figure}

\subsection{Secret key rate formula} 
The secret key rate formula for the protocol above has been derived in~\cite{koji} when the number of multiplexing $m$ tends to infinity. It reads:

\begin{equation}\label{skr}
R=p_{\rm{Z}}^2\,p_{\rm{s}}\,p_{\rm{BSM}}\,[1-f\,h(e_{\rm{Z}})-h(e_{\rm{X}})],
\end{equation} 
where $p_{\rm{s}}$ is the probability that Charlie's QND and the $Z$ measurements are both successful either at Alice's or Bob's site. The quantity $p_{\rm{BSM}}$ represents the success probability of one BSM. The quantities $e_{\rm{Z}}$ and $e_{\rm{X}}$, on the other hand denote the bit and phase error rates, respectively. The parameter $f$ is an inefficiency function for the error correction process and $h(x)=-x\log_2(x)-(1-x)\log_2(1-x)$ is the binary entropy function. We remark that the fact that in Eq.~\eqref{skr} only $p_{\rm{s}}$ rather than the square of it appears is due to the advantage that the original AMDI-QKD protocol offers, that is, BSMs are only performed between signals that survived the channel loss. Therefore, a particular signal only needs to survive travelling through one path~\cite{koji}.

We also note that since we evaluate the secret key rate in the asymptotic regime where $m\rightarrow\infty$, the parties perform steps 1-5 of the protocol only once. Also, for simplicity, we assume that $p_{\rm{Z}}\gg p_{\rm{X}}$, so that we take $p_{\rm{Z}}^2\approx 1$ in Eq.~\eqref{skr} for the simulations below. Moreover, we assume for simplicity that $f=1$.

The full-mode analysis of the AMDI-QKD protocol described above can be found in Appendix \ref{appendix:FMA}, based on Eq.~\eqref{skr} and the device models described in the next section.

\section{Device models}\label{Sec:ToolBox}

Here, we describe the mathematical models that we use to characterise the behavior of the different devices employed to evaluate the performance of the AMDI-QKD protocol.

\subsection{Photonic sources}
  
We shall assume that all entanglement sources emit states of the following form:

\begin{equation}\label{source}
\ket{\psi}= \sum_{n=0}^{\infty}\,\sqrt{p_n}\, \ket {\phi_n},
\end{equation}
where $\sum_{n=0}^{\infty}p_n=1$ and $p_n\geq 0$. The $n$-photon pair states $\ket {\phi_n}$ are given by
\begin{equation}\label{statesEmittedByTheSources}
\sket {\phi_n} =\, \frac{1}{n!\,\sqrt{n+1}}(x_H^{\dagger}y_H^{\dagger}+x_V^{\dagger}y_V^{\dagger})^n\sket{0},
\end{equation}
where $x_H^{\dagger}$ and $y_H^{\dagger}$ ($x_V^{\dagger}$ and $y_V^{\dagger}$) are the creation operators of horizontally (vertically) polarized photons in the modes $x$ and $y$, respectively and $\sket{0}$ denotes the vacuum state.

We note that if we choose $p_1=1$ (\ie, $p_n=0$ for any $n\neq 1$) then Eq.~\eqref{source} represents a perfect entanglement source that is capable of emitting maximally entangled states $\ket {\phi_1}$ with certainty. Another interesting special case is when

\begin{equation}\label{PDCgeneral}
p_n=\frac{(n+1)\lambda^n}{(1+\lambda)^{n+2}},
\end{equation}
holds, in which case Eq.~\eqref{source} describes a type-II PDC source~\cite{PDC} with $\lambda$ being a positive parameter, related to the amplitude of the pumping laser.

\subsection{Detectors}

We shall assume that all the detectors are PNR detectors, characterized by the following positive operator-valued measure (POVM) elements 
\begin{equation}\label{PNRPOVM}
\Pi_{k}= \sum_{n=k}^{\infty} \binom{n}{k}\eta_{\rm{det}}^k (1-\eta_{\rm{det}})^{n-k} \proj{n},
\end{equation}
with $k=0,1,\cdots,\infty$ denoting the number of detected photons. In Eq.~\eqref{PNRPOVM}, the parameter $\eta_{\rm{det}}$ is the detection efficiency of the detectors and $\ket{n}$ is the $n$-photon Fock state. We note that for simplicity in Eq.~\eqref{PNRPOVM} we have disregarded the dark count probability of the PNR detectors.
\subsection{QND measurement}
A linear-optics teleportation-based implementation of the QND measurement can be seen in Fig.~\ref{QndLayout}. It consists of a standard linear-optics BSM module together with an entanglement source $S^{\rm{QND}}$ that emits the state described in Eq.~\eqref{source}.
The purpose of the QND measurement module is to herald if a photon successfully arrived at Charlie's node, so that, in this case, he can continue the protocol with performing his BSMs. The successful heralding events are constituted by the same detection patterns as in the original MDI-QKD protocol~\cite{mdi}, that is, a successful heralding event occurs if the detectors detect exactly one photon in horizontal polarization and one photon in vertical polarization. The setup is based on quantum teleportation. Indeed, in the case of a single-photon input and $S^{\rm{QND}}$ being a perfect EPR source with $p_1=1$, the QND module teleports its input mode to its output mode~\cite{teleportation}. We note that there exist more efficient implementations of a BSM~\cite{GriceBSM, ewertBSM} in terms of the success probability (which can approach $1$ instead of $1/2$ as in the scheme shown in Fig.~\ref{QndLayout}), but they come with the overhead of the need for complicated ancilla states.
\begin{figure}[H]
\centering
  {\includegraphics[scale=0.4]{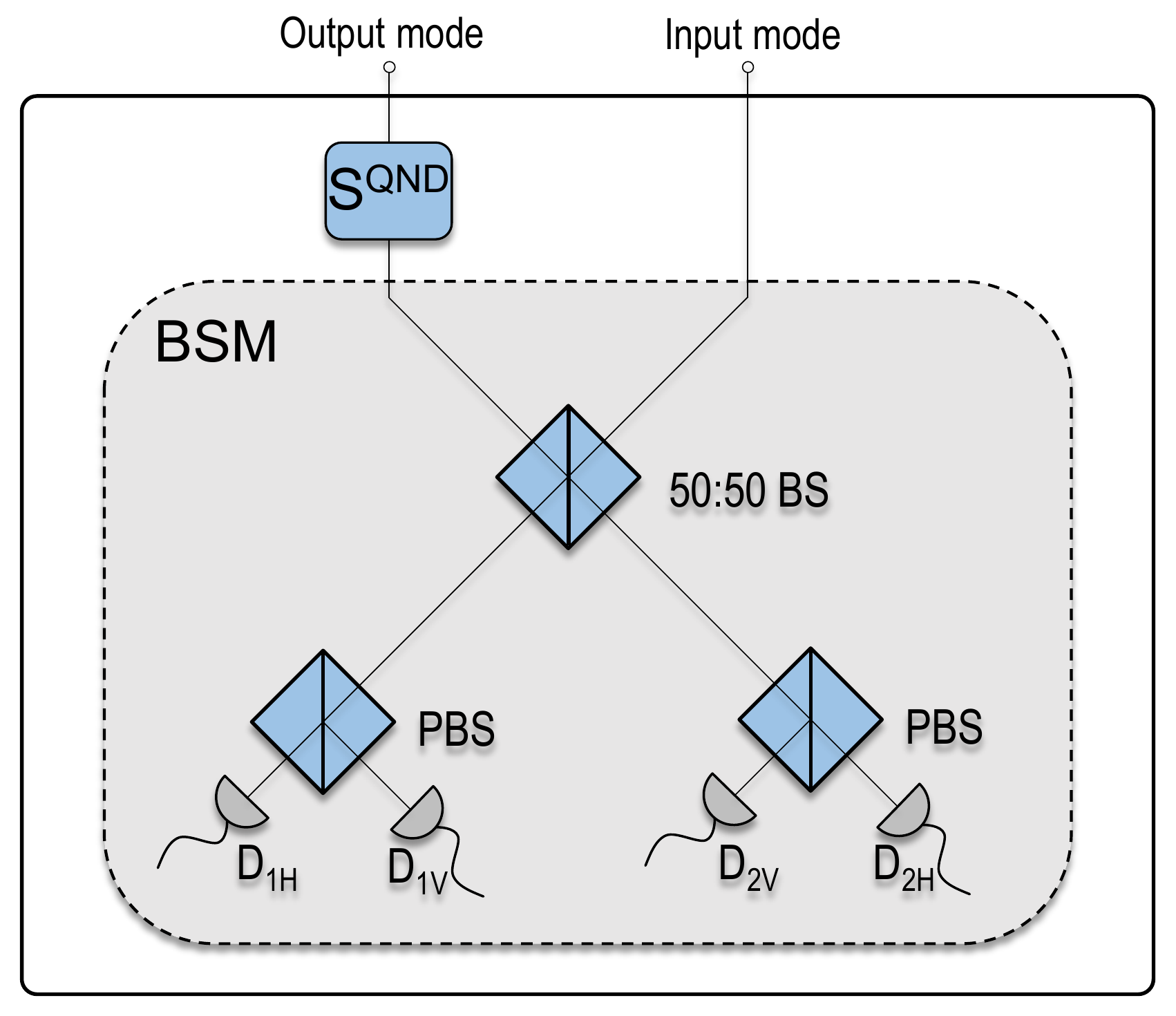}}
  \caption{Linear-optics teleportation-based implementation of a QND measurement. The symbol $S^{\rm{QND}}$ denotes an entanglement source. The optical devices involved are 50:50 beam splitters (BS), polarizing beam splitters (PBS) and PNR detectors ($D_{\rm{1H}}$, $D_{\rm{1V}}$, $D_{\rm{2H}}$, $D_{\rm{2V}}$) that project into horizontal ($H$) or vertical ($V$) polarization. A successful detection event, heralding the arrival of a photon, corresponds to observing exactly one photon in $H$ polarization and one photon in $V$ polarization (\ie, altogether two photons). If $D_{\rm{1H}}$ and $D_{\rm{2V}}$, or $D_{\rm{1V}}$ and $D_{\rm{2H}}$, detect one photon each, this corresponds to a projection into the singlet state $\ket{\psi^-}=1/\sqrt 2(\ket{HV}-\ket{VH})$, while if $D_{\rm{1H}}$ and $D_{\rm{1V}}$, or $D_{\rm{2H}}$ and $D_{\rm{2V}}$, detect one photon each, that corresponds to a projection into the triplet state $\ket{\psi^+}=1/\sqrt 2(\ket{HV}+\ket{VH})$. In the case of a successful detection event, the state of the input photon is teleported to the output mode (up to a unitary transformation).}
\label{QndLayout}
\end{figure} 

\subsection{Optical switches}

Charlie directs the successfully arriving signals into his BSMs with the help of optical switches. For this, an active feedforward mechanism is required since it takes time to align the ports of the switches properly, so that successful signals  from Alice and Bob end up in the same BSM at Charlie's site. It is assumed that one active feedforward takes time $\tau$, during which signals propagate in optical fibers. Therefore, the feedforward procedure can be modelled as a lossy channel with transmittance $\eta_{\rm{f}}=\exp(-\tau c /L_{\rm{att}})$ due to the propagation of the signals through the fibers, where $c$ denotes the speed of light in the optical fiber. Otherwise, we assume perfect switching devices with unlimited number of entries.

\subsection{BSM modules after the switches}

In this case, the linear-optics implementation of the middle BSM modules is depicted in Fig.~\ref{BsmLayout}. Similarly to the QND module, the success events are constituted by the same detection patterns as in the original MDI-QKD protocol~\cite{mdi}. Note, however, that this middle BSM module differs from the one used in the QND measurement. Here, Hadamard gates (denoted in Fig.~\ref{BsmLayout} by H) are applied~\cite{HadamardOperation}, whose operation can be described as
\begin{equation}\label{HadamardOperation}
\begin{bmatrix}
    y_H^{\dagger} \\[0.1cm]
    y_V^{\dagger} 
\end{bmatrix}=\frac {1} {\sqrt{2}}
\begin{bmatrix}
    1 & 1  \\[0.15cm]
    1 & -1 
\end{bmatrix}\\
\begin{bmatrix}
    x_H^{\dagger} \\[0.1cm]
    x_V^{\dagger} 
\end{bmatrix},
\end{equation}
where $y_H^{\dagger}$ ($y_V^{\dagger}$) and $x_H^{\dagger}$ ($x_V^{\dagger}$) are the creation operators of the input and output modes of the Hadamard gate in horizontal (vertical) polarization.

\begin{figure}[H]
\centering
  {\includegraphics[scale=0.43]{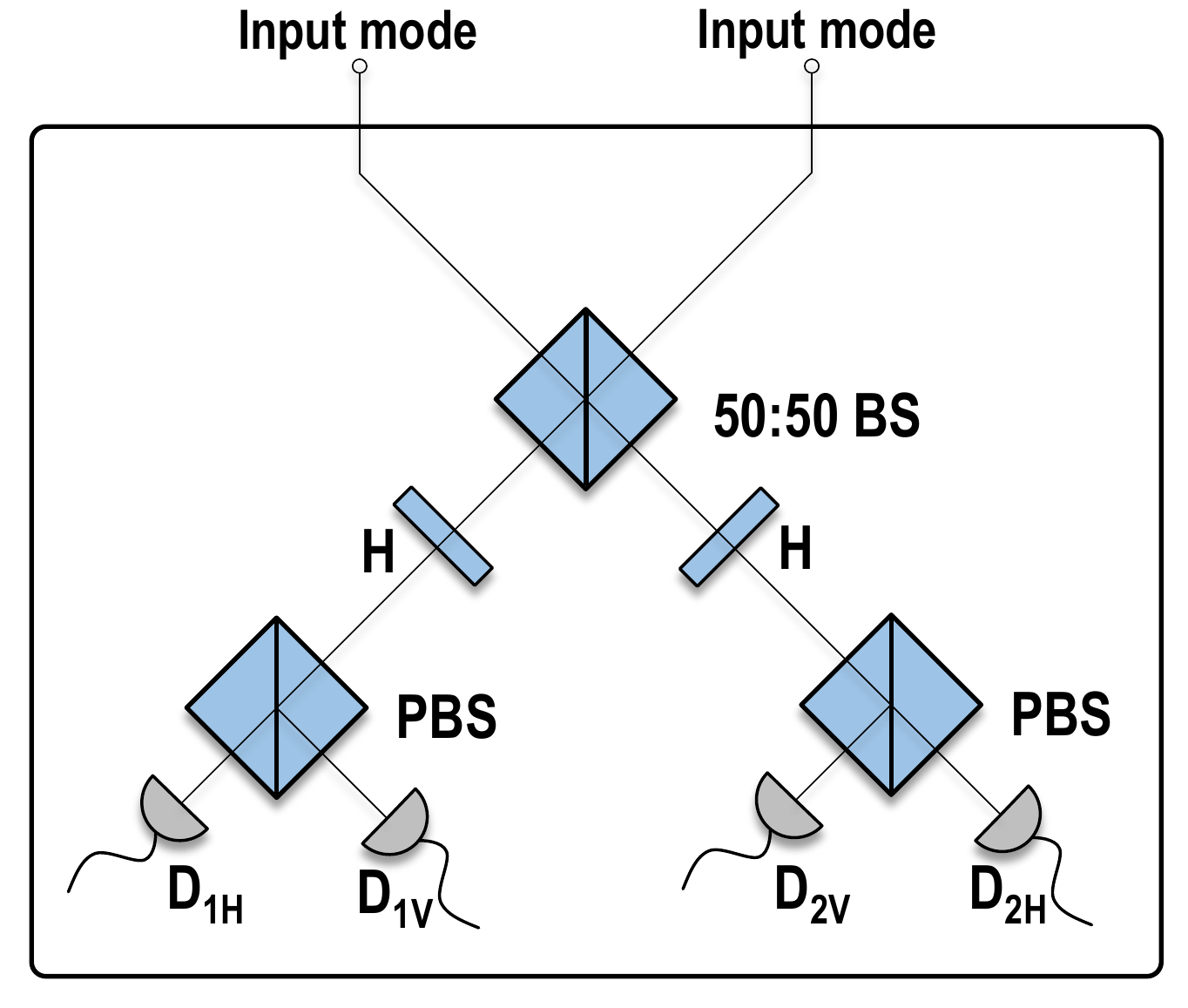}}
  \caption{Linear-optics implementation of the middle BSMs after the switches. The optical devices involved are 50:50 beam splitters (BS), Hadamard gates (H), polarizing beam splitters (PBS) and PNR detectors ($D_{\rm{1H}}$, $D_{\rm{1V}}$, $D_{\rm{2H}}$, $D_{\rm{2V}}$) that project into horizontal ($H$) or vertical ($V$) polarization. A successful detection event corresponds to observing exactly one photon in $H$ polarization and one photon in $V$ polarization (\ie, altogether two photons). If $D_{\rm{1H}}$ and $D_{\rm{2V}}$, or $D_{\rm{1V}}$ and $D_{\rm{2H}}$, detect one photon each, this corresponds to a projection into the singlet state $\ket{\psi^-}=1/\sqrt 2(\ket{HV}-\ket{VH})$, while if $D_{\rm{1H}}$ and $D_{\rm{1V}}$, or $D_{\rm{2H}}$ and $D_{\rm{2V}}$, detect one photon each, that corresponds to a projection into the state $\ket{\phi^-}=1/\sqrt 2(\ket{HH}-\ket{VV})$.}
\label{BsmLayout}
\end{figure}
Including the Hadamard gates is advantageous, since they can prevent errors that can occur otherwise. For example, when the source $S^{\rm{QND}}$ emits a two-photon pair state and the input of the QND module is the vacuum state, a successful heralding event can still occur, despite of the fact that the signal coming from the corresponding party was lost. In this scenario, it can be shown (see Appendix \ref{appendix:ExplanationForHad}) that the output state of the QND measurement, heading towards the middle BSM, is a two photon signal consisting of one vertically and one horizontally polarized photon. Now, without the Hadamard gates, this state could give a successful detection event at the middle BSM, which is undesired since there was no signal received from the party connected to the seemingly successful QND measurement. Intuitively, one expects that adding the Hadamard gates, \ie, rotating the signal, will just decrease the chance of the two photons ending up in detectors corresponding to different polarizations. This is so because, after the rotation, it is not predetermined any more which detector they hit after the polarizing beam splitter, and it rather becomes a probabilistic process. However, carrying out the calculation shows (see Appendix \ref{appendix:ExplanationForHad}) that, in fact, the Hadamard gates force the spurious two photon signal to give non-conclusive detection events at the final BSM, which is the best possible scenario.

\subsection{Quantum channel}
For concreteness, we assume that the quantum channel connecting Alice (Bob) to Charlie is an optical fiber with attenuation length $L_{\rm{att}}$ and transmittance \begin{equation}
\eta_{\rm{ch}}=\exp\left[-\frac{L}{2L_{\rm{att}}}\right],
\end{equation} where $L$ is the distance between Alice and Bob, and $L_{\rm{att}}=10\,\textrm{dB}/(\alpha\log 10)$, where $\log$ denotes the natural logarithm and $\alpha$ is the loss coefficient of the channel measured in dB/km. For simplicity, we do not consider any misalignment effect in our system.

\subsection*{}
Using the above devices, we can summarize the differences between the AMDI-QKD scheme considered in this paper and the original proposal. In particular, here we have replaced both the perfect entanglement source for the QND measurement and the perfect single-photon sources of Alice and Bob in the original AMDI-QKD~\cite{koji} by more realistic entanglement sources, described by Eq.~\eqref{source}, which have non-zero probability of emitting multiple-photon signals in general. Note that the single-photon sources possessed by Alice and Bob are achieved in our scheme by measuring one mode of their entanglement sources. Another modification is to change all the threshold detectors, which were in the BSM and the QND measurement, to PNR detectors. Apart from these, the protocol, of course, runs in a similar manner as in the original proposal~\cite{koji}, described in Sec.~\ref{Sec:steps}.

\section{Necessary conditions}\label{Sec:NecessaryConditions}

In this section, we first derive a simple but non-trivial analytical necessary condition on the photon-number statistics of the sources for the AMDI-QKD protocol to be able to beat the repeaterless bound reported in~\cite{PLOB}, which is given by the formula $-\log_2(1-\eta_{\rm{ch}}^2)$. For this, we use the fact that, if one cannot beat the bound with $\eta_{\rm{det}}=1$ and $\tau=0\,\rm{s}$, then it is also impossible to beat it with 
$\eta_{\rm{det}}<1$ and $\tau>0\,\rm{s}$, since the secret key rate cannot increase with lower efficiency detectors. Thus, we can construct the necessary condition by requiring that the bound in~\cite{PLOB} is overcome when we set $\eta_{\rm{det}}=1$ and $\tau=0\,\rm{s}$, which significantly simplifies the secret key rate formula, making the derivation of an analytical result possible. With this simple necessary condition, it is already possible to show that PDC sources cannot beat the repeaterless bound with the AMDI-QKD protocol.

Afterwards, we consider the general case where $\eta_{\rm{det}}<1$ and $\tau>0\,\rm{s}$. In this scenario we can only obtain the condition on the sources to beat the repeaterless bound by numerically evaluating the secret key rate formula, given by Eq.~\eqref{skr}. The required formulas to evaluate Eq.~\eqref{skr} can be found in Appendix~\ref{appendix:FMA}. 

By comparing the two results, we see that the finite detection efficiency of the detectors has a significant impact on the required photon-number statistics of the sources. To be precise, when the detection efficiency of the detectors decreases, then the requirements on the maximum values of the multi-photon probabilities of the sources become more severe. 
 
\subsection{Simple analytical necessary condition}

When we set $\eta_{\rm{det}}=1$ and $\tau=0\,\rm{s}$, we find that the secret key rate formula can be written as (for the details of the calculation see Appendix \ref{appendix:FMA}):
\begin{equation}\label{skrFormulaUnitEff}
R_{[\eta_{\rm{det}}=1,\,\tau=0]}=\frac{3\,p_1\,q_1^2\,\eta_{\rm{ch}}^2}{8\,q_2+4\,(3\,q_1-2\,q_2)\,\eta_{\rm{ch}}},
\end{equation}
where the $p_n$ ($q_m$) values represent the photon-number statistics of the sources $S_{\rm{AC}}$ and $S_{\rm{BC}}$ ($S^{\rm{QND}}$). We remark that we can choose the values of $p_n$ to be equal for the sources $S_{\rm{AC}}$ and $S_{\rm{BC}}$ because Charlie is located halfway between them, and such selection is optimal in this scenario. Importantly, we note that in Eq.~\eqref{skrFormulaUnitEff} only the probabilities $p_1$, $q_1$ and $q_2$ appear. This is so because of the following. The appearance of $p_1$ and $q_1$ is the consequence of the fact that only single-photon pairs can give rise to true success events. Using unit efficiency detectors, the only possible way for multiple-photon pairs to cause a seemingly successful QND measurement is when the $S^{\rm{QND}}$ source emits the state $\ket {\phi_2}$ and the signal from Alice (Bob) is lost during the transmission through the optical fiber, which explains the appearance of $q_2$.  

From Eq.~\eqref{skrFormulaUnitEff} we learn, that in the case of the use of detectors with unit efficiency, in the asymptotic limit of $L\rightarrow\infty$ (\ie, $\eta_{\rm{ch}}\rightarrow 0$), $R_{[\eta_{\rm{det}}=1,\,\tau=0]}\simeq 3\,p_1\,q_1^2/(8\,q_2)\,\eta_{\rm{ch}}^2$, which means that choosing the parameters $p_1$, $q_1$ and $q_2$ properly, we can actually beat the repeaterless bound~\cite{PLOB} in this limit. The reason why the secret key rate never breaks down in the unit efficiency detector case is due to the fact that we have neglected dark counts and misalignment in the quantum channel.

\begin{claim*}
To beat the repeaterless bound~\cite{PLOB} with the AMDI-QKD protocol, using PNR detectors and entanglement sources $S_{\rm{AC}}$ and $S_{\rm{BC}}$ ($S^{\rm{QND}}$) of the form given by Eq.~\eqref{source}, with photon-number statistics $p_n$ ($q_m$), it is necessary that \begin{equation}\label{FinalConditionWithValidityCheck}
q_2\leq \min \left(\frac{25}{96}\,p_1\,q_1^2\,, 1-q_1 \right), 
\end{equation} is fulfilled.
\end{claim*}

\begin{proof}
Making use of the Taylor expansion, the following inequality trivially holds for the repeaterless bound, reported in~\cite{PLOB}:
\begin{equation}\label{PLOBTaylor}
-\log_2(1-\eta_{\rm{ch}}^2)=\sum_{n=1}^{\infty}\frac{\eta_{\rm{ch}}^{2n}}{n\ln2}\geq 1.44 \eta_{\rm{ch}}^2.
\end{equation}
Then, for the necessary condition, we require that $R_{[\eta_{\rm{det}}=1,\,\tau=0]}$, given by Eq.~\eqref{skrFormulaUnitEff}, is greater than $1.44\,\eta_{\rm{ch}}^2$.

If $(3q_1-2q_2)\leq0$ holds for the sources, an upper bound can be given on $R_{[\eta_{\rm{det}}=1,\,\tau=0]}$ by plugging $\eta_{\rm{ch}}=1$ in the denominator of Eq.~\eqref{skrFormulaUnitEff}:
\begin{align}\label{skrneg}
&R_{[\eta_{\rm{det}}=1,\,\tau=0]}=\frac{3\,p_1\,q_1^2\,\eta_{\rm{ch}}^2}{8\,q_2+4\,(3\,q_1-2\,q_2)\,\eta_{\rm{ch}}}\leq \frac{3\,p_1\,q_1^2\,\eta_{\rm{ch}}^2}{12\,q_1}\nonumber\\
&=\frac{p_1\,q_1\eta_{\rm{ch}}^2}{4}\leq\frac {\eta_{\rm{ch}}^2}{4}\leq 1.44\eta_{\rm{ch}}^2.
\end{align}
Therefore, to be able to overcome the repeaterless bound~\cite{PLOB}, $(3q_1-2q_2)>0$ must hold, which means that we can upper bound the secret key rate by plugging $\eta_{\rm{ch}}=0$ in the denominator of Eq.~\eqref{skrFormulaUnitEff}:

\begin{equation}\label{skrpos}
R_{[\eta_{\rm{det}}=1,\,\tau=0]}=\frac{3\,p_1\,q_1^2\,\eta_{\rm{ch}}^2}{8\,q_2+4\,(3\,q_1-2\,q_2)\,\eta_{\rm{ch}}}\leq \frac{3\,p_1\,q_1^2\,\eta_{\rm{ch}}^2}{8\,q_2}.
\end{equation}
For the necessary condition it is then required that

\begin{equation}\label{neccond}
\frac{3\,p_1\,q_1^2\,\eta_{\rm{ch}}^2}{8\,q_2}\geq  1.44\,\eta_{\rm{ch}}^2,
\end{equation}
holds, which can be simplified to
\begin{equation}\label{neccondfinal}
q_2\leq \frac{25\,p_1\,q_1^2}{96} .
\end{equation}
This inequality is tighter than $(3q_1-2q_2)>0$. Moreover, since $\sum_{m=0}^{\infty}q_m=1$, we have that 
\begin{equation}\label{validityCondition}
q_2\leq1-q_1.
\end{equation}
Considering this validity condition, we obtain the simple necessary condition given by Eq.~\eqref{FinalConditionWithValidityCheck}.
\end{proof}

\begin{figure}[H]
\centering
  {\includegraphics[scale=0.5]{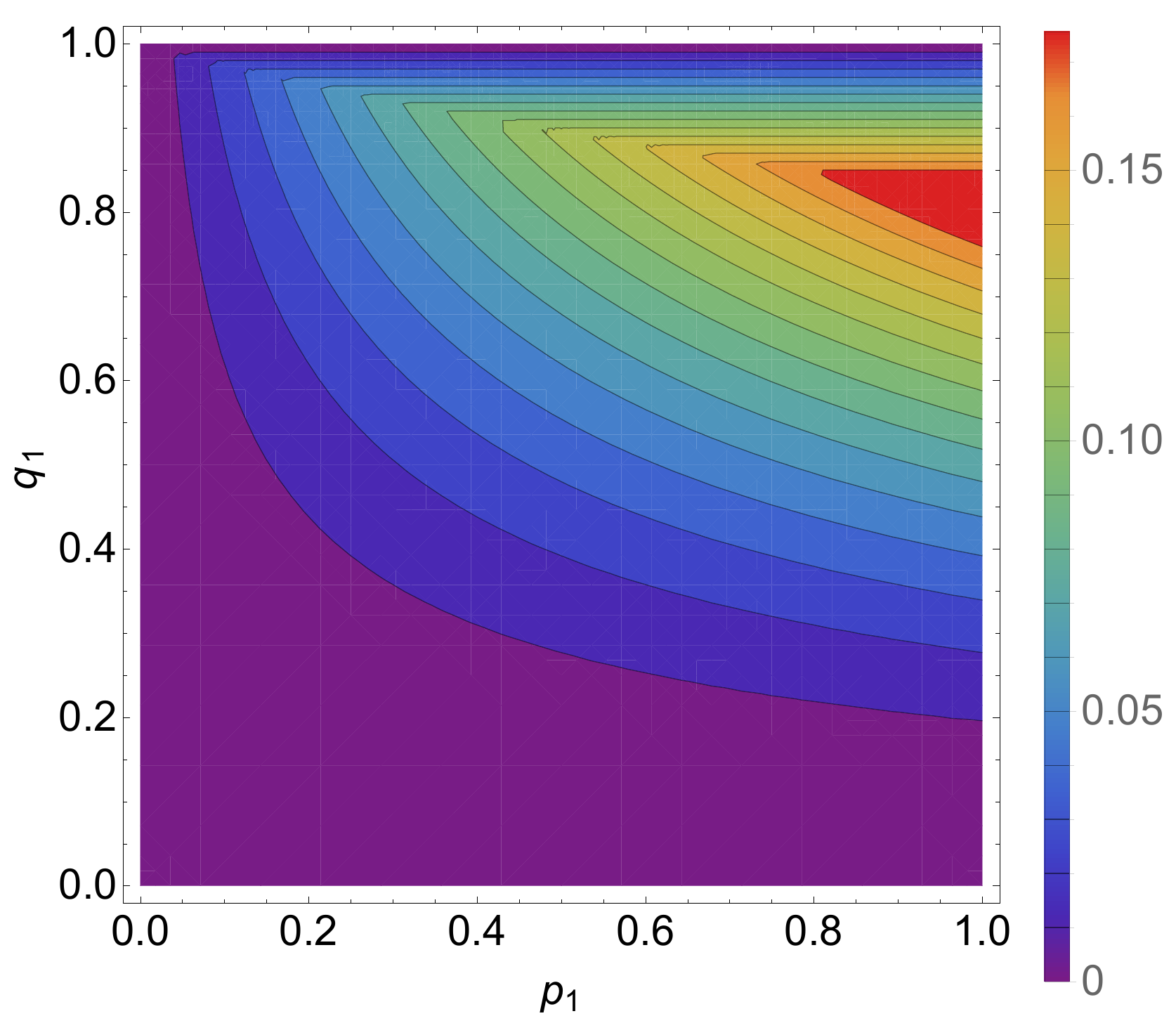}}
  \caption{The maximal allowable value of $q_2$ as a function of $p_1$ and $q_1$, which is necessary for the sources to be able to beat the repeaterless bound~\cite{PLOB} with the AMDI-QKD protocol.}
  \label{Plot:necCondition}
\end{figure}

The necessary condition, given by Eq.~\eqref{FinalConditionWithValidityCheck}, is depicted in Fig.~\ref{Plot:necCondition}, where we plot the value of $q_2$, above which it is not possible to beat the bound~\cite{PLOB}, given the values of $p_1$ and $q_1$. We can see that the necessary condition is more sensitive to the value of $q_1$, than to the value of $p_1$ in the sense that if $q_1$ is too small then no matter how large we set $p_1$, we need to have $q_2=0$ (purple area in Fig.~\ref{Plot:necCondition}). However, in the converse situation, having $p_1$ small does not imply that $q_2=0$. One can also see this formally by noting that in Eq.~\eqref{FinalConditionWithValidityCheck} only $q_1$ appears squared. This is so, because, as explained before, even in the ideal efficiency detector case the source $S^{\rm{QND}}$ can cause errors by introducing two-photon pair signals. While the sources $S_{\rm{AC}}$ and $S_{\rm{BC}}$ cannot, since the detectors in the $Z$/$X$ measurement filter out all the two-photon pair signals.

\begin{cor*}
Using PDC sources, it is impossible to beat the repeaterless bound~\cite{PLOB} with the AMDI-QKD protocol. 
\end{cor*}

\begin{proof}
The photon-number statistics of the PDC sources can be written as
\begin{equation}\label{PDClambda}
p^{\rm{PDC}}_n=\frac{(n+1)\lambda^n}{(1+\lambda)^{n+2}}\,\,\rm{and}\,\,\,\textit{q}^{\rm{PDC}}_{\it{m}}=\frac{(\it{m}+\rm{1})\mu^{\it{m}}}{(1+\mu)^{\it{m}+\rm{2}}},
\end{equation}
where, as already mentioned previously, $\lambda$ ($\mu$) is a positive parameter, related to the amplitude of the laser used to pump the sources $S_{\rm{AC}}$ and $S_{\rm{BC}}$ ($S^{\rm{QND}}$). We note that due to the symmetries of the setup, that is, Charlie is located halfway between Alice and Bob, we can set the $\lambda$ values for the sources $S_{\rm{AC}}$ and $S_{\rm{BC}}$ equal.
 
Plugging Eq.~\eqref{PDClambda} into Eq.~\eqref{FinalConditionWithValidityCheck}, we have that
\begin{equation}\label{forcontradiction}
\frac{\lambda}{(1+\lambda)^3\,(1+\mu)^2}\geq \frac{36}{25}
\end{equation}
is necessary to overcome the repeaterless bound~\cite{PLOB}. However, since $(1+\mu)^{-2}<1$ and $\max\{\lambda\,(1+\lambda)^{-3}\}= 4/27$ (the latter can be seen easily by taking the derivative with respect to $\lambda$), we have that   
\begin{equation}\label{contradiction}
\frac{\lambda}{(1+\lambda)^3\,(1+\mu)^2}\leq \frac{4}{27},
\end{equation}
which obviously contradicts Eq.~\eqref{forcontradiction}, meaning that for PDC sources the necessary condition cannot be fulfilled. Therefore, it is impossible to overcome the repeaterless bound~\cite{PLOB}.
\end{proof}

\subsection{Tighter necessary condition}

In this section, we now analyse the necessary condition without the assumption of $\eta_{\rm{det}}=1$ and $\tau=0\,\rm{s}$. For this, we numerically evaluate the general secret key rate formula derived in Appendix~\ref{appendix:FMA}.

For simplicity, however, in the simulations, we restrict ourselves to the case where every source emits at most two photon pairs, that is, we set $q_n=p_n=0$ for any $n\geq 3$. We remark, however, that, with the general secret key rate formula given in Appendix~\ref{appendix:FMA}, it is possible to allow for an arbitrary number of emitted photon pairs if one has sufficient computational power. The results can be seen in Fig.~\ref{Plot:Comparison}. We introduced the quantities $P\equiv p_2/p_1$ and $Q\equiv q_2/q_1$ to characterise the quality of the sources, lower values meaning higher quality sources since emitting two photon pairs is less likely than emitting the desired EPR pair, for given values of $p_0$ and $q_0$.

Let $Q^{\rm{max}}_{[\eta_{\rm{det}},\tau]}$ denote the maximally allowable value of $Q$ to be able to overcome the bound~\cite{PLOB}, as a function of $p_0$ and $P$, given the parameters $\eta_{\rm{det}}$ and $\tau$. In Fig.~\ref{Plot:Comparison}, we plot the quantity $Q^{\rm{max}}_{[\eta_{\rm{det}}<1,\, \tau=67\,\rm{ns}]}/Q^{\rm{max}}_{[\eta_{\rm{det}}=1, \,\tau=0\,\rm{ns}]}$ as a function of $p_0$ and $P$ for different values of $\eta_{\rm{det}}<1$, given the value of $q_0$. So that we can observe how does the $Q^{\rm{max}}$ value decreases for a given $p_0$ and $P$ if we have $\eta_{\rm{det}}<1$ and $\tau=67\,\rm{ns}$, compared to the $\eta_{\rm{det}}=1$ and $\tau=0\,\rm{ns}$ case. We set $\tau=67\,\rm{ns}$ as illustration, since this is the value used in the original implementation~\cite{koji}, based on the experiments reported in~\cite{feedforward1,feedforward2}. Our simulations show that the value of $q_0$ does not influence significantly the order of magnitude tendencies observed while decreasing the detection efficiency of the detectors. Loosely speaking, varying the value of $q_0$ only translates the secret key rate curve vertically (\ie, the secret key rate basically decreases by a constant factor for all values of the channel loss) for given $P$ and $p_0$ values. This means that in the comparison, every curve for the different detection efficiencies will be translated by the same factor, therefore, for a different $q_0$ value, the $Q^{\rm{max}}$ values will also be altered by a common factor for each detection efficiency, and since we are plotting their quotients, it means that the plots on Fig.~\ref{Plot:Comparison} will stay very similar.

\begin{figure}[H]
  \centering  
  \subfigure[]{\includegraphics[scale=0.33]{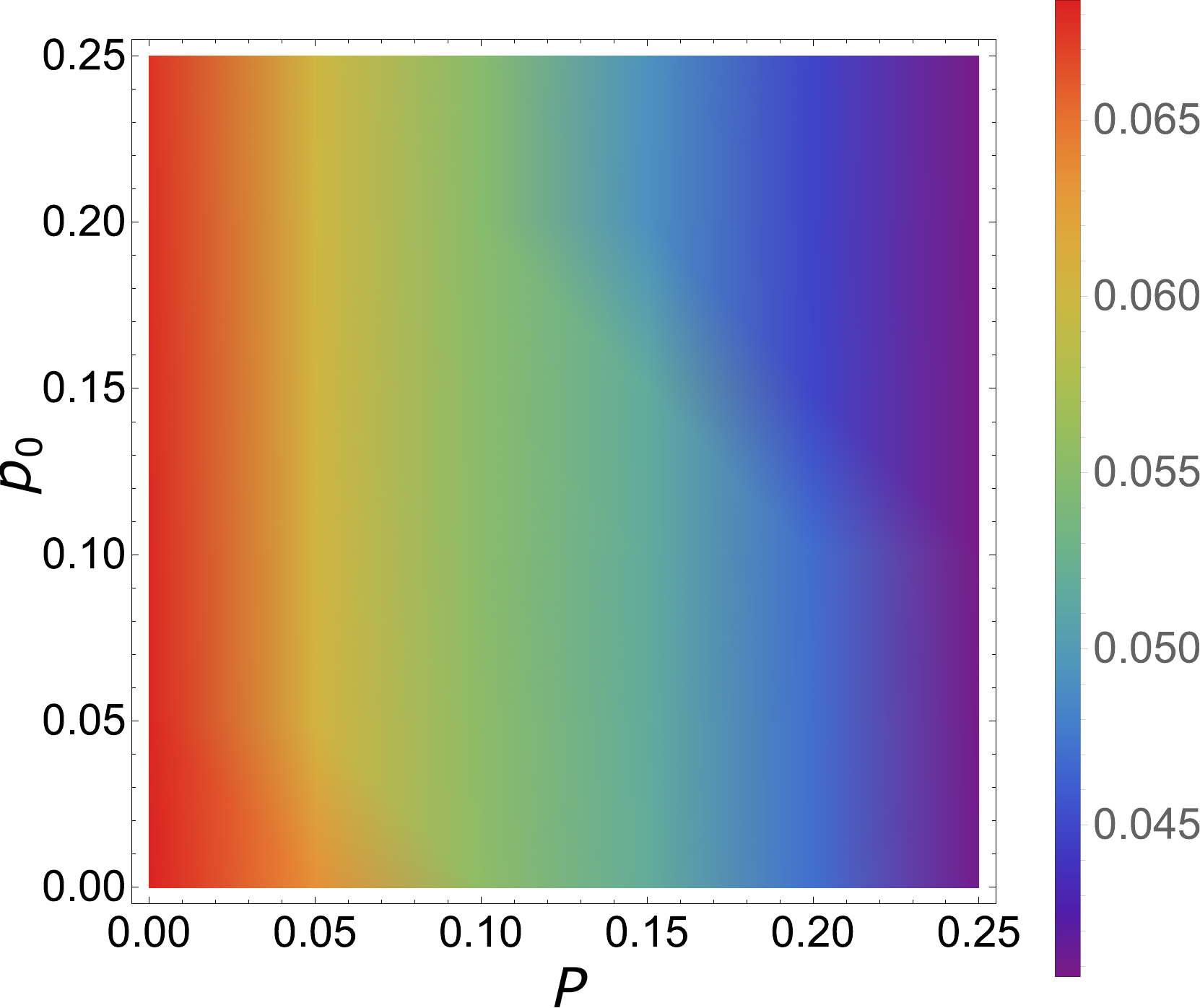}}
  \subfigure[]{\includegraphics[scale=0.33]{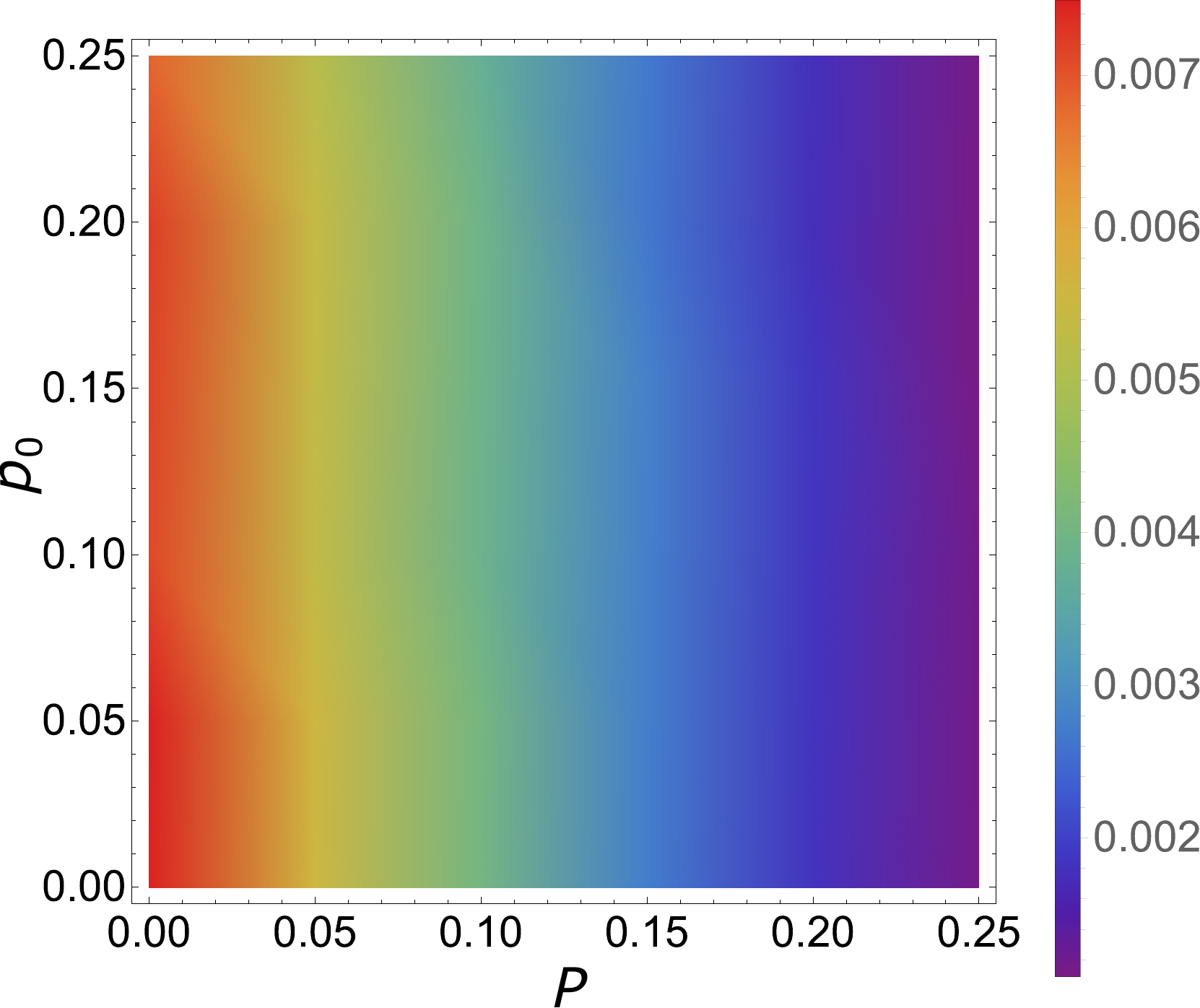}} 
  \subfigure[]{\includegraphics[scale=0.33]{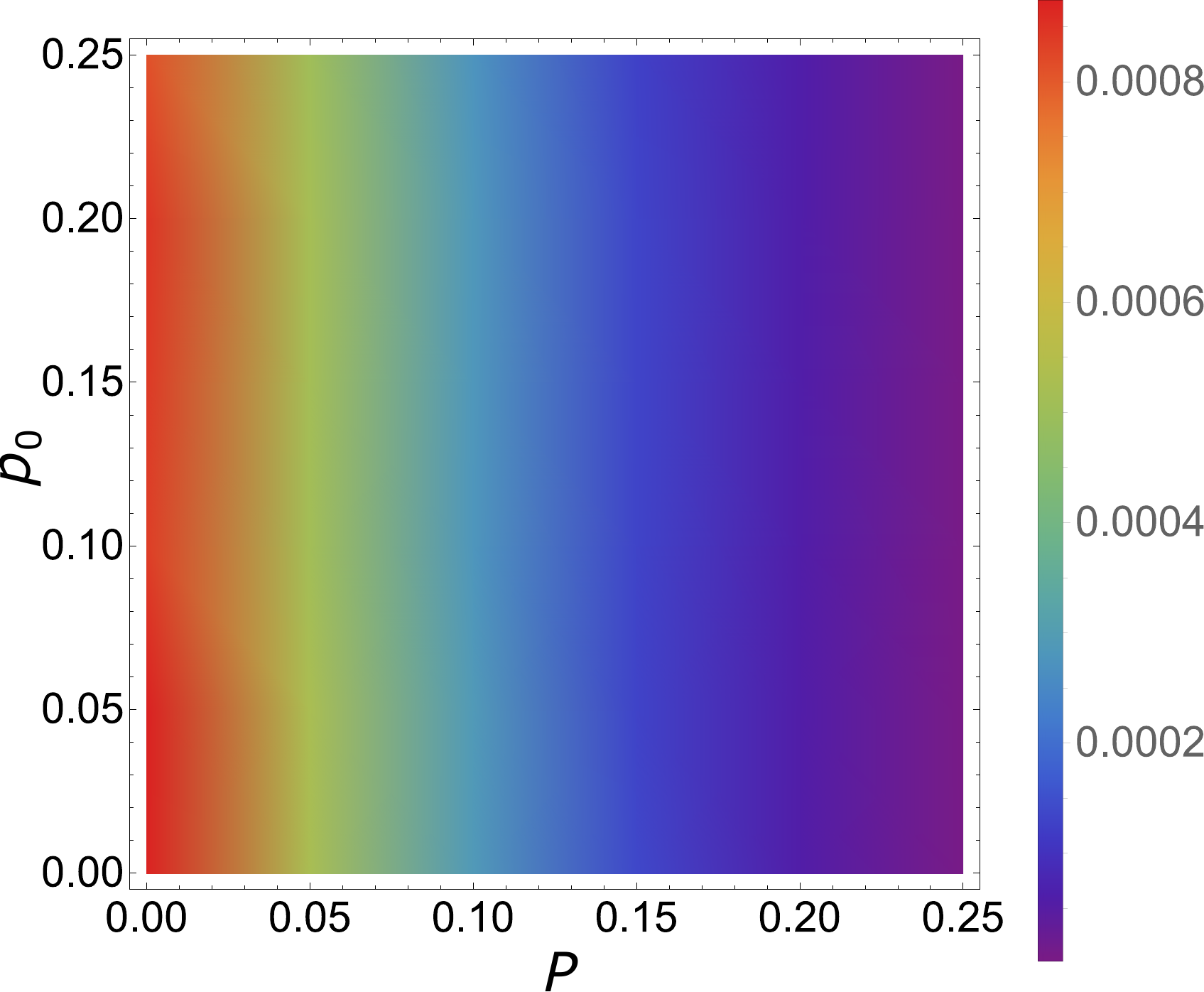}}
  \caption{Comparison of the analytically obtained necessary condition using the unit efficiency assumption and $\tau = 0\,\rm{ns}$ to the numerically obtained necessary condition when $\tau = 67\,\rm{ns}$~\cite{koji,feedforward1,feedforward2} and (a) $\eta_{\rm det}=0.9$, (b) $\eta_{\rm det}=0.7$ and (c) $\eta_{\rm det}=0.5$. In each figure, the quantity plotted is $Q^{\rm{max}}_{[\eta_{\rm{det}},\,67\,\rm{ns}]}/Q^{\rm{max}}_{[1,\,0\,\rm{ns}]}$ as a function of $p_0$ and $P\equiv p_2/p_1$. The quantity $Q^{\rm{max}}_{[\eta_{\rm{det}},\tau]}$ denotes the maximally allowable value of $Q\equiv q_2/q_1$ in order to overcome the bound~\cite{PLOB}. Moreover, for the simulations, we assume that $L_{\rm{att}}=22\,\rm{km}$, $c=2\cdot 10^8\,\rm{m/s}$ and $q_0=0.2$. See the main text for further details.}
  \label{Plot:Comparison}
\end{figure}

First, let us examine the region of the plots where $P\approx 0$ (reddish areas in Fig.~\ref{Plot:Comparison}), which means that the sources of Alice and Bob are close to perfect entanglement sources. We see that decreasing the value of $\eta_{\rm{det}}$ from 1 to 0.9, 0.7 and 0.5 results in the need for about an order of magnitude higher quality sources at Charlie's hand for each efficiency compared to the previous one. This is partly due to fact that eight detections are needed for obtaining a raw key bit (one in the $Z$/$X$ measurement of each Alice and Bob, two in the QND measurement at each side and two more in the BSM), therefore, now we have a factor of $\eta_{\rm{det}}^8<1$ in the probability of each key generation event, but this only translates the secret key rate curve vertically. On top of this, and what is more important, we now have an increased probability of obtaining erroneous key generation events from the two-photon pair component of the sources $S^{\rm{QND}}$. To see this, suppose that one of the sources $S^{\rm{QND}}$ on Alice's side emits a two-photon pair signal and all the other sources emit one-photon pair signals, and suppose that the signal from Alice towards the QND measurement is lost in the transmission. In this case, as already mentioned previously, we can get a seemingly successful detection event if one photon out of the two going from the QND module towards the BSM is lost in the detection process of the BSM, which happens with a probability $1-\eta_{\rm{det}}$. Thus, the error rate will increase as we decrease the detection efficiency. The only way to compensate this error is to have better quality sources at Charlie's hand (meaning lower $Q$ values). 

Now, let us observe the part of the plots, where $P\approx 0.25$ (purplish areas in Fig.~\ref{Plot:Comparison}), meaning that Alice and Bob no longer have perfect entanglement sources. In this case, we find that the more we decrease the value of $\eta_{\rm{det}}$ the more the values of $Q^{\rm{max}}_{[\eta_{\rm{det}}<1,\, \tau=67\,\rm{ns}]}/Q^{\rm{max}}_{[\eta_{\rm{det}}=1, \,\tau=0\,\rm{ns}]}$ in the region of $P\approx 0.25$ will decrease compared to the values in the region of $P\approx 0$. In other words, in the $P\approx 0.25$ region the $Q^{\rm{max}}_{[\eta_{\rm{det}}<1,\, \tau=67\,\rm{ns}]}/Q^{\rm{max}}_{[\eta_{\rm{det}}=1, \,\tau=0\,\rm{ns}]}$ values do not decrease linearly with $\eta_{\rm{det}}$, which is more or less true in the $P\approx 0$ region. What has been said before for the region $P\approx 0$ is true for this region as well, but, on top of those effects, since $P>0$, the sources of Alice and Bob now have a non-zero probability of producing erroneous successful detections in the $Z$/$X$ measurements due to their two-photon pair component, which was not possible before in the $P\approx 0$ region. Consider the previously explained situation with the difference, that Alice's source now emits a two-photon pair signal and both photons from Alice heading towards the QND module are lost in the transmission. In this case we can only get a seemingly successful raw key generation event if one photon out of the two in the $Z$/$X$ measurement at Alice's site is lost in the detection, which occurs with a probability $1-\eta_{\rm{det}}$. Thus, the probability of this type of error scales with $(1-\eta_{\rm{det}})^2$ (the other $1-\eta_{\rm{det}}$ factor comes from the fact that in the scenario considered, we need to lose one more photon in the detection process of the BSM to have a seemingly successful raw key generation event), meaning that it does not depend linearly on the detection efficiency as in the $P\approx 0$ region, which explains the observed behaviour for the $P\approx 0.25$ region. From this, we can see that, as expected, $P$ has a more significant impact on the $Q^{\rm{max}}$ values than $p_0$, which, since increasing $p_0$ will not increase the probabilities of an error, just translates the secret key rate curve vertically. These are the main reasons behind this dramatic increase in the quality of Charlie's sources as we decrease $\eta_{\rm{det}}$.

Summing up the observations from Fig.~\ref{Plot:Comparison} we can say that the necessary quality of the sources can be much higher than what is expected from the unit detection efficiency condition if we have non-perfect detectors. We also note that using the formulas for the secret key rate from Appendix~\ref{appendix:FMA}, Fig.~\ref{Plot:Comparison} can be easily made for arbitrary $0\leq q_0<1$ and $0<\eta_{\rm{det}}\leq 1$.

\section{Conclusion}\label{Sec:Conclusion}

We have investigated the performance of a more realistic implementation of the original AMDI-QKD protocol~\cite{koji}, assuming that the parties have access to a broad class of entanglement sources of the form given by Eq.~\eqref{source} and photon-number resolving (PNR) detectors. We have shown that the improved $\eta_{\rm{ch}}$ scaling (with $\eta_{\rm{ch}}$ being the transmittance of the channel connecting Alice and Bob to Charlie), offered by the protocol, is very sensitive to multiple-photon pair components emitted by the sources.

More precisely, we have derived a simple non-trivial analytical necessary condition on the photon-number statistics of the entanglement sources to be able to overcome the repeaterless bound~\cite{PLOB} with the AMDI-QKD protocol. With this condition, we have demonstrated analytically that employing the widely available parametric down-conversion sources does not enable the protocol to beat the repeaterless bound. Furthermore, we have quantitatively investigated the effect that the finite detection efficiency of the detectors have on the required photon-number statistics of the sources. In this regard, we have shown that, when the detection efficiency of the detectors decreases, then the maximum tolerable values of the multi-photon probabilities of the sources in order to beat the repeaterless bound become significantly more severe. This latter study, however, was only feasible by numerically evaluating the secret key rate formula of the protocol.

Our results suggest, that, while the AMDI-QKD protocol could in principle overcome the repeaterless bound with idealized devices, in practice it demands very high quality entanglement sources, which are thus still challenging to realize with current technology, besides, of course, the involved multiplexing techniques depending on the channel length.

\section{Acknowledgement}
We thank the Spanish Ministry of Economy and Competitiveness (MINECO), the Fondo Europeo de Desarrollo Regional (FEDER) through the grant TEC2017-88243-R, and the European Union's Horizon 2020 research and innovation programme under the Marie Sklodowska-Curie grant agreement No 675662 (project QCALL) for financial support. K.A. thanks support, in part, from
PRESTO, JST JPMJPR1861.

\appendix
\section{BSM with Hadamard gates}
\label{appendix:ExplanationForHad}
In this appendix we give an intuitive argument on the reason why the use of Hadamard gates is advantageous in the BSM after the optical switches. Suppose that the state used for key generation is $\ket{\phi_1}$, given by Eq.~\eqref{statesEmittedByTheSources}. If every source emits this state and a successful detection pattern occurs, the parties share the desired quantum correlation to obtain their secret key (given that dark counts are neglected). However, if some of the sources emit multiple-photon pair states, for example the state $\ket{\phi_2}$, also given by Eq.~\eqref{statesEmittedByTheSources}, this could result in a seemingly successful detection event, which does not provide the parties with the desired correlations and will end up producing errors. As we will show below, including the Hadamard gates in the BSM removes the possibility of getting errors from the state $\ket{\phi_2}$.

\begin{figure}[H]
\centering
  {\includegraphics[scale=0.29]{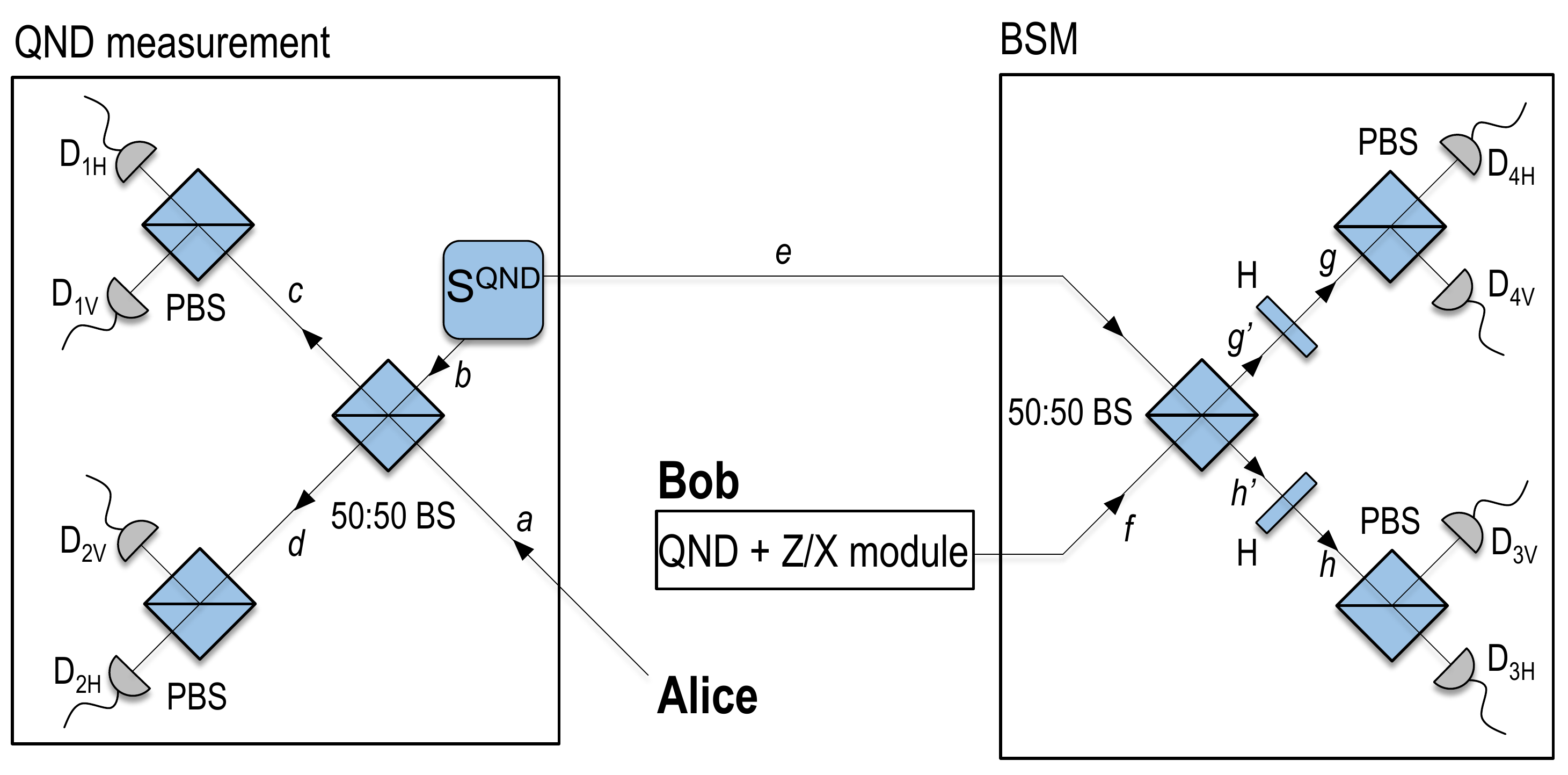}}
  \caption{Schematic of the optical modes used in Appendix~\ref{appendix:ExplanationForHad}. We use 50:50 beam splitters (BS), PNR detectors ($D_{\rm{1H}}$, $D_{\rm{1V}}$, $D_{\rm{2H}}$, $D_{\rm{2V}}$, $D_{\rm{3H}}$, $D_{\rm{3V}}$, $D_{\rm{4H}}$, $D_{\rm{4V}}$)  with detection efficiency $\eta_{\rm{det}}$  and polarizing beam splitters (PBS) that project into horizontal ($H$) or vertical ($V$) polarization. Optical modes are denoted by italic letters.}
  \label{Plot:ExplanationForHad}
  \centering
\end{figure}

Let us illustrate this with an example depicted in Fig.~\ref{Plot:ExplanationForHad}. In particular, let us consider the situation when $S_{\rm{AC}}$ emits the state $\ket{\phi_1}$ and there is a correct detection in Alice's $Z$/$X$ measurement but the photon in the other mode (mode $a$ in Fig.~\ref{Plot:ExplanationForHad}), going towards the QND measurement, is lost in the transmission. Moreover, suppose that the $S^{\rm{QND}}$ source on Alice's side emits the $\ket{\phi_2}$ state, which can cause a seemingly successful QND measurement on Alice's side and also a seemingly successful BSM at Charlie's site.

Furthermore, suppose that the sources $S_{\rm{BC}}$ and $S^{\rm{QND}}$ on Bob's side both emit the $\ket{\phi_1}$ state. Let us assume that these signals cause successful $Z$/$X$ and QND measurements on Bob's side, but the photon in the other mode (mode $f$ in Fig.~\ref{Plot:ExplanationForHad}), going from the $S^{\rm{QND}}$ source towards Charlie's BSM module, is lost, say during the feedforward mechanism. Therefore, in this example, there is no actual signal coming from Bob to the BSM. This way, if the Hadamard gates are not used, it is possible that the $\ket{\phi_2}$ state coming from the $S^{\rm{QND}}$ source on Alice's side will result in a seemingly successful BSM. Consequently, the parties would conclude that the protocol had run correctly and they can obtain a secret key bit, but in reality, they will get random outcomes instead of correlated ones. Next, we show that this cannot occur in the presence of the Hadamard gates.

For simplicity, instead of $\ket{\phi_2}$, we assume that the source $S^{\rm{QND}}$ on Alice's side emits the following unnormalised state
\begin{align}\label{unnormstates}
&\ket{\phi'_2}=2\sqrt 3 \ket{\phi_2}=(b_H^{\dagger}e_H^{\dagger}+b_V^{\dagger}e_V^{\dagger})^2\sket{0}\nonumber\\
&=\left[(b_H^{\dagger})^2(e_H^{\dagger})^2+(b_V^{\dagger})^2(e_V^{\dagger})^2+2 b_H^{\dagger}b_V^{\dagger}e_H^{\dagger}e_V^{\dagger}\right]\sket{0}.
\end{align}

The QND measurement is successful if there are exactly two, orthogonally polarized photons in the modes $c$ and $d$ (see Fig.~\ref{Plot:ExplanationForHad}):

\begin{align}\label{QNDsuccessPattern}
&c_H^{\dagger} d_V^{\dagger}=\frac 12 \left(a_H^{\dagger} a_V^{\dagger}  +  a_V^{\dagger} b_H^{\dagger}  -  a_H^{\dagger} b_V^{\dagger}  -  b_H^{\dagger} b_V^{\dagger} \right) ,
\nonumber\\
&c_H^{\dagger} c_V^{\dagger}=\frac 12  \left(a_H^{\dagger} a_V^{\dagger}   +  a_V^{\dagger} b_H^{\dagger}   +  a_H^{\dagger} b_V^{\dagger}   +  b_H^{\dagger} b_V^{\dagger}\right),
\nonumber\\
&c_V^{\dagger} d_H^{\dagger}= \frac 12  \left(a_H^{\dagger} a_V^{\dagger}   -  a_V^{\dagger} b_H^{\dagger}   +  a_H^{\dagger} b_V^{\dagger}   -  b_H^{\dagger} b_V^{\dagger}\right) , 
\nonumber\\
&d_H^{\dagger} d_V^{\dagger}=\frac 12  \left(a_H^{\dagger} a_V^{\dagger}   -  a_V^{\dagger} b_H^{\dagger}   -  a_H^{\dagger} b_V^{\dagger}   +  b_H^{\dagger} b_V^{\dagger}\right),
\end{align}
where we expressed the output modes $c$ and $d$ after the 50:50 BS in the QND measurement as a function of the input modes $a$ and $b$. As explained above, in the particular example considered, there is no photon in mode $a$. Therefore, a successful event could only be caused by the $b_H^{\dagger} b_V^{\dagger}$ component from Eq.~\eqref{QNDsuccessPattern}. Comparing with $\ket{\phi'_2}$ in Eq.~\eqref{unnormstates}, we conclude, that if the QND measurement on Alice's side succeeded, then the state in mode $e$ has to be characterised by $e_H^{\dagger}e_V^{\dagger}$. 

As explained previously, there is no photon in mode $f$ (coming from Bob). This means that, since there is one horizontally and one vertically polarized photon incident on the 50:50 BS in the BSM (mode $e$), a successful detection pattern in the BSM is easily produced, if there are no Hadamard gates.

Now, let us consider what happens with the Hadamard gates. Similarly to the QND measurement, the BSM is successful if there are exactly two, orthogonally polarized photons in the modes $g$ and $h$:

\begin{align}\label{BSMsuccessPattern1}
g_H^{\dagger} g_V^{\dagger}&=\frac{(e_H^{\dagger})^2}{4} - \frac{(e_V^{\dagger})^2}4 + \frac{e_H^{\dagger} f_H^{\dagger}}2 + \frac{(f_H^{\dagger})^2}4 - \frac{e_V^{\dagger} f_V^{\dagger}}2 
\nonumber\\&- \frac{(f_V^{\dagger})^2}4,
\end{align}
\begin{align}\label{BSMsuccessPattern2}
g_H^{\dagger} h_V^{\dagger}&=\frac{(e_H^{\dagger})^2}{4} - \frac{(e_V^{\dagger})^2}4 - \frac{e_V^{\dagger} f_H^{\dagger}}2- \frac{(f_H^{\dagger})^2}4+ \frac{e_H^{\dagger} f_V^{\dagger}}2 \nonumber\\&+ \frac{(f_V^{\dagger})^2}4,
\end{align}
\begin{align}\label{BSMsuccessPattern3} 
h_V^{\dagger} h_H^{\dagger}&= \frac{(e_H^{\dagger})^2}4 - \frac{(e_V^{\dagger})^2}4 - \frac{e_H^{\dagger} f_H^{\dagger}}2 + \frac{(f_H^{\dagger})^2}4 + \frac{e_V^{\dagger} f_V^{\dagger}}2 \nonumber\\&- \frac{(f_V^{\dagger})^2}4,
\end{align}
\begin{align}\label{BSMsuccessPattern4}
g_V^{\dagger} h_H^{\dagger}&=\frac{(e_H^{\dagger})^2}4 - \frac{(e_V^{\dagger})^2}4 + \frac{e_V^{\dagger} f_H^{\dagger}}2 - \frac{(f_H^{\dagger})^2}4 - \frac{e_H^{\dagger} f_V^{\dagger}}2 \nonumber\\&+ \frac{(f_V^{\dagger})^2}4,
\end{align}
where we expressed the modes $g$ and $h$ after the Hadamard gates as a function of the input modes $e$ and $f$. In doing so, we used the matrices describing the quantum optical operation of the BS and the Hadamard gate. It is clear that in this example, the state $e_H^{\dagger}e_V^{\dagger}$, which we have after the seemingly successful QND measurement, cannot cause a successful BSM, since in Eqs.~\eqref{BSMsuccessPattern1}-\eqref{BSMsuccessPattern4} there is no component that contains $e_H^{\dagger}e_V^{\dagger}$. Therefore, the error that was possible before is now filtered out by the Hadamard gates.

\section{Full-mode analysis}
\label{appendix:FMA} 

\subsection{Rephrasing the secret key rate formula}\label{App:rephrase}
The secret key rate formula, given by Eq.~\eqref{skr}, can be written as~\cite{koji}
\begin{equation}\label{skrAppendix}
R=\,p_{\rm{s}}\,p_{\rm{BSM}}\,\left[1-h(e_{\rm{Z}})-h(e_{\rm{X}})\right],
\end{equation}
where we have set $f=1$. Also, as explained in the main text, we assume that $p_Z\approx 1$ since we consider the asymptotic scenario. We remind the reader that $p_{\rm{s}}$ is the probability that Charlie's QND and the $Z$/$X$ measurements are both successful either at Alice's or Bob's site. The quantity $p_{\rm{BSM}}$ represents the success probability of one BSM. We note that since the $Z$-basis is used for the key generation the above quantities are defined in the case when Alice and Bob choose the $Z$-basis. We also note that the probabilities $p_{\rm{s}}$ and $p_{\rm{BSM}}$ by definition include all the possible detection patterns that constitute that particular success event.

However, due to the symmetries of the channel model, for our simulations, it is not necessary to calculate the probabilities of all the detection patterns that constitute a certain event, which would be rather tedious and redundant. Thus, in the remainder of this section, we are going to relate the quantities $p_{\rm{s}}$, $p_{\rm{BSM}}$, $e_{\rm{Z}}$ and $e_{\rm{X}}$ to the probabilities of some particular detection patterns, relying on the symmetries of the channel model.

\begin{figure}[H]
\centering
  {\includegraphics[scale=0.362]{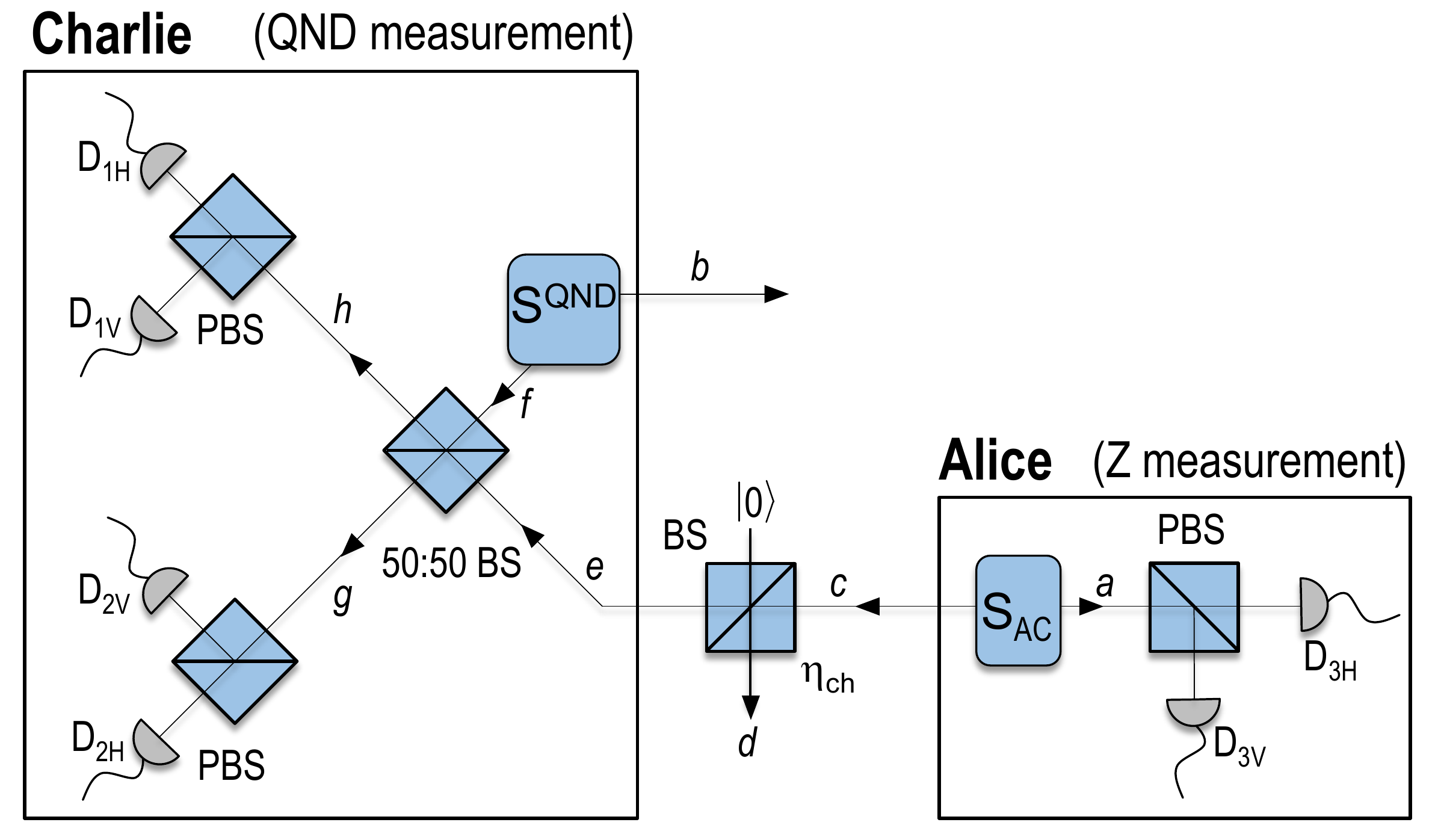}}
  \caption{Layout of the QND and $Z$ measurements. We use 50:50 beam splitters (BS), PNR detectors ($D_{\rm{1H}}$, $D_{\rm{1V}}$, $D_{\rm{2H}}$, $D_{\rm{2V}}$, $D_{\rm{3H}}$, $D_{\rm{3V}}$) with detection efficiency $\eta_{\rm{det}}$ and polarizing beam splitters (PBS) that project into horizontal ($H$) or vertical ($V$) polarization. The quantum channel is modeled by a BS with transmittance $\eta_{\rm{ch}}=\exp(-L/(2L_{\rm{att}}))$, where $L_{\rm{att}}$ is the attenuation length of the optical fiber and $L/2$ is the distance between Alice (Bob) and Charlie. The entanglement sources are denoted by $S_{\rm{AC}}$ and $S^{\rm{QND}}$. Note, that if Alice (Bob) chooses to measure in the $X$-basis, then a Hadamard gate is applied to mode $a$. Optical modes are denoted by italic letters.}
  \label{Plot:QNDZX}
  \centering
\end{figure}

First, let $p_{\rm{QND}}$ denote the probability that there is exactly one photon detected in each of the detectors $D_{\rm{2H}}$, $D_{\rm{2V}}$ and $D_{\rm{3H}}$ and zero photons detected in the other detectors $D_{\rm{1H}}$, $D_{\rm{1V}}$ and $D_{\rm{3V}}$ in Fig.~\ref{Plot:QNDZX}. Note that this means a successful QND measurement and simultaneously a successful $Z$ measurement on Alice's side, where she obtained the horizontal ($H$) polarization, and therefore this particular detection pattern represents one of the patterns that constitute $p_{\rm{s}}$. A successful QND measurement can be realized by four different detection patterns (observing altogether two photons in the QND module, one in $H$ polarization and one in $V$ polarization, \ie, if $D_{\rm{1H}}$ and $D_{\rm{2V}}$, or $D_{\rm{1V}}$ and $D_{\rm{2H}}$, or $D_{\rm{1H}}$ and $D_{\rm{1V}}$, or $D_{\rm{2H}}$ and $D_{\rm{2V}}$ in Fig.~\ref{Plot:QNDZX} detect one photon each). The $Z$ measurement can be realized by two different detection patterns (one photon detected in either $D_{\rm{3H}}$ or $D_{\rm{3V}}$ in Fig.~\ref{Plot:QNDZX}). Altogether, this means eight different possibilities. Note that these eight detection patterns all have the same probability due to the symmetries of the channel model considered. Moreover, $p_{\rm{QND}}$ is also independent of the basis choice of Alice. Therefore, we can write that 
\begin{equation}\label{ps}
p_{\rm{s}}=8\,p_{\rm{QND}}.
\end{equation}

Now, let us also express $p_{\rm{BSM}}$ with probabilities corresponding to particular detection patterns. For this, first, let $p_{\rm{c}}^{\rm{Z}}$ (correct) denote the probability of the following particular detection pattern given that the parties choose to measure their local modes in the $Z$-basis. Suppose that Charlie's QND and the $Z$ measurement were successful on both Alice's and Bob's side with the particular detection pattern described before for $p_{\rm{QND}}$ (\ie, both parties detected $H$ polarization) and then in the BSM after the optical switches there is exactly one photon detected in each of the detectors $D_{\rm{2H}}$ and $D_{\rm{2V}}$ and zero photons detected in the other detectors $D_{\rm{1H}}$ and $D_{\rm{1V}}$ in Fig.~\ref{Plot:BSM}. Note, that this detection pattern corresponds to a projection into the Bell state $\ket{\phi^-}$, which means that the parties will not apply bit flip, this way obtaining correlated (correct) raw key bits. The success probability of the BSM alone, corresponding to the above described particular detection pattern can be written as $p_{\rm{c}}^{\rm{Z}}/p_{\rm{QND}}^2$. We also note that in the BSM after the switches there is another detection pattern that corresponds to obtaining correlated (correct) raw key bits (projection into the Bell state $\ket{\phi^-}$), which happens when there is one photon detected in both $D_{\rm{1H}}$ and $D_{\rm{1V}}$ and zero photons detected in $D_{\rm{2H}}$ and $D_{\rm{2V}}$ in Fig.~\ref{Plot:BSM}. Due to the symmetries, this particular click pattern will also have probability $p_{\rm{c}}^{\rm{Z}}$. Therefore, the contribution to $p_{\rm{BSM}}$ in this case will be $2\,p_{\rm{c}}^{\rm{Z}}/p_{\rm{QND}}^2$.

\begin{figure}[H]
\centering
  {\includegraphics[scale=0.4]{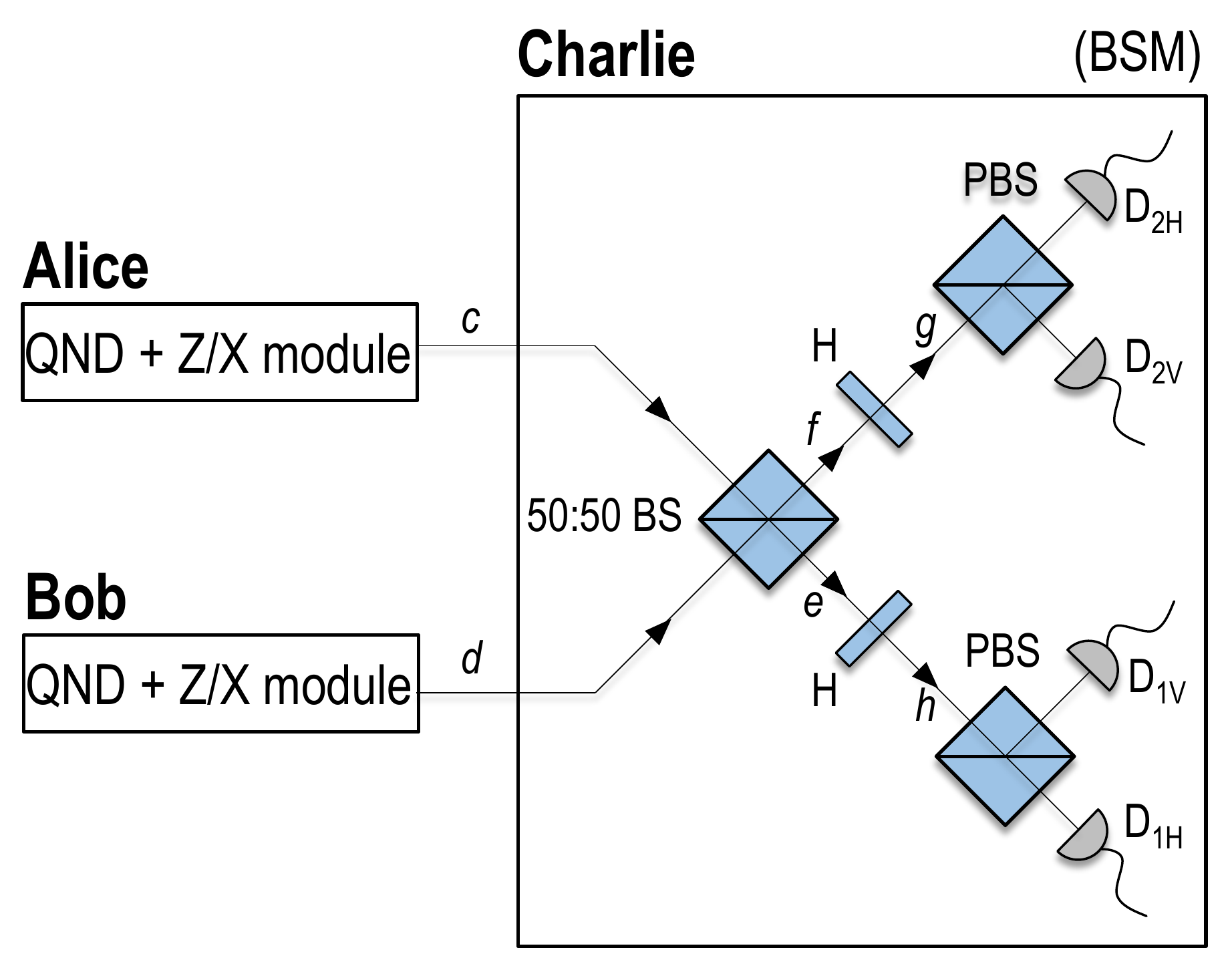}}
  \caption{Layout for the BSM after the optical switches. We use 50:50 beam splitters (BS), Hadamard gates (H), PNR detectors ($D_{\rm{1H}}$, $D_{\rm{1V}}$, $D_{\rm{2H}}$, $D_{\rm{2V}}$) with detection efficiency $\eta'_{\rm{det}}=\eta_{\rm{det}}\eta_{\rm{f}}$ and polarizing beam splitters (PBS) that project into horizontal ($H$) or vertical ($V$) polarization. Here $\eta_{\rm{f}}$ is the channel loss that corresponds to the active feedforward mechanism and in our calculations it has been incorporated into the detection efficiency. Optical modes are denoted by italic letters.}
  \label{Plot:BSM}
  \centering
\end{figure}

Now, let us define $p_{\rm{nc}}^{\rm{Z}}$ (non-correct) in the same way as we defined $p_{\rm{c}}^{\rm{Z}}$, with the only difference being that in the BSM after the optical switches there is exactly one photon detected in each of the detectors $D_{\rm{2H}}$ and $D_{\rm{1V}}$ and zero photons detected in the other detectors $D_{\rm{1H}}$ and $D_{\rm{2V}}$ in Fig.~\ref{Plot:BSM}. Note that this particular detection pattern corresponds to a projection into the Bell state $\ket{\psi^-}$, meaning that one of the parties will apply a bit flip, therefore obtaining anti-correlated (non-correct) raw key bits. Similarly to $p_{\rm{c}}^{\rm{Z}}$, there is another detection pattern that corresponds to obtaining the same anti-correlated raw key bits (one photon detected in both $D_{\rm{1H}}$ and $D_{\rm{2V}}$ and zero photons detected in $D_{\rm{1V}}$ and $D_{\rm{2H}}$ in Fig.~\ref{Plot:BSM}). Again, due to the symmetries, this particular click pattern will also have $p_{\rm{nc}}^{\rm{Z}}$ probability, therefore the contribution to $p_{\rm{BSM}}$ is $2\,p_{\rm{nc}}^{\rm{Z}}/p_{\rm{QND}}^2$ in this case.

The quantity $e_{\rm{X}}$ can be defined and calculated, similarly to the case for $e_{\rm{Z}}$, by considering probabilities with which agreed and disagreed bits are adopted by Alice and Bob. However, for clarity, here we use another method to calculate $e_{\rm{X}}$.
Choosing the $X$-basis means that the parties apply a Hadamard gate on the mode $a$ in Fig.~\ref{Plot:QNDZX} before their local measurement with detectors $D_{\rm{3H}}$ and $D_{\rm{3V}}$.
From the symmetry of the protocol, without loss of generality, we can focus on a specific success event of Charlie where  
Charlie's QND measurements on Alice's side and Bob's side announce single-photon detection in each of the detectors $D_{{\rm 2H}}$ and $D_{{\rm 2V}}$ and zero-photon detection in the other detectors $D_{{\rm 1H}}$ and $D_{{\rm 1V}}$ in Fig.~\ref{Plot:QNDZX}, and Charlie's final Bell measurement announces single-photon detection in each of the detectors $D_{{\rm 2H}}$ and $D_{{\rm 2V}}$ and zero-photon detection in the other detectors  $D_{{\rm 1H}}$ and $D_{{\rm 1V}}$ in Fig.~\ref{Plot:BSM}. 
Then, since Alice and Bob would share entanglement close to $\ket{\phi^-}$, we define $p_{\rm{c}}^{\rm{X}}$ (correct) [$p_{\rm{nc}}^{\rm{X}}$ (non-correct)] as a probability with which Charlie obtains the specific success event and Alice (Bob) detects exactly one photon in the detector $D_{\rm{3H}}$ ($D_{\rm{3V}}$) [$D_{\rm{3H}}$ ($D_{\rm{3H}}$)] on her (his) side and zero photons in the other detector included in her (his) $X$ measurement in Fig.~\ref{Plot:QNDZX}.

With these at our hands, we can write that
\begin{equation}
p_{\rm{BSM}}=\frac{2\left(p_{\rm{c}}^{\rm{Z}}+p_{\rm{nc}}^{\rm{Z}}\right)}{p_{\rm{QND}}^2}
\end{equation} 
and
\begin{equation}
e_{\rm{Z}}=\frac{p_{\rm{nc}}^{\rm{Z}}}{p_{\rm{c}}^{\rm{Z}}+p_{\rm{nc}}^{\rm{Z}}},\qquad e_{\rm{X}}=\frac{p_{\rm{nc}}^{\rm{X}}}{p_{\rm{c}}^{\rm{X}}+p_{\rm{nc}}^{\rm{X}}}.
\end{equation}
Consequently, we can rewrite Eq.~\eqref{skrAppendix} to the following form:
\begin{equation}\label{AdequateSkrAppendix}
R=\frac{16\,(p_{\rm{c}}^{\rm{Z}}+p_{\rm{nc}}^{\rm{Z}})}{p_{\rm{QND}}}\left[1-h\left(\frac{p_{\rm{nc}}^{\rm{Z}}}{p_{\rm{c}}^{\rm{Z}}+p_{\rm{nc}}^{\rm{Z}}}\right)-h\left(\frac{p_{\rm{nc}}^{\rm{X}}}{p_{\rm{c}}^{\rm{X}}+p_{\rm{nc}}^{\rm{X}}}\right)\right].
\end{equation}
The remainder of Appendix~\ref{appendix:FMA} is structured as follows. In Appendix~\ref{App:QND} we derive $p_{\rm{QND}}$, then in Appendix~\ref{App:pcAndpnc} and Appendix~\ref{App:pcAndpncX} we derive $p_{\rm{c}}^{\rm{Z}}$, $p_{\rm{nc}}^{\rm{Z}}$ and $p_{\rm{c}}^{\rm{X}}$, $p_{\rm{nc}}^{\rm{X}}$. And finally, in~\ref{app:Unit} we obtain the secret key rate formula for the 
$\eta_{\rm{det}}=1$ and $\tau=0\,\rm{s}$ case.

\subsection{Derivation of $p_{\rm{QND}}$}\label{App:QND}

The layout for this derivation can be seen in Fig.~\ref{Plot:QNDZX}. For convenience, we use the density matrix formalism for the sources. It is easy to see that converting the emitted state $\ket{\psi}$, given by Eq.~\eqref{source}, into a density matrix in the $\sket {\phi_n}$-basis, only the diagonal terms will give contributions when calculating the probabilities of the different detection patterns (described by the POVMs of Eq.~\eqref{PNRPOVM}, which are diagonal in the Fock basis) since different values of $n$ represent different photon numbers. This means that for our calculations, instead of using the pure states given by Eq.~\eqref{source}, we can use mixed states of the following form:

\begin{align}\label{sources}
\rho_{\rm{AC}}&= \sum_{n=0}^{\infty}\,p_n\, \ketbra {\phi_n} {\phi_n},\nonumber \\
\rho^{\rm{QND}}&= \sum_{m=0}^{\infty}\,q_m\, \ketbra {\varphi_m} {\varphi_m},
\end{align}
with $\sum_{n=0}^{\infty}p_n=1$ and $\sum_{m=0}^{\infty}q_m=1$. For the experimental setup considered, the results in both cases (\ie, by using Eq.~\eqref{source} or Eq.~\eqref{sources}) coincide. The states $\ket {\phi_n}$ and $\ket {\varphi_m}$ are given by
\begin{align}\label{states}
\sket {\phi_n} &=\, \frac{1}{n!\,\sqrt{n+1}}(a_H^{\dagger}c_H^{\dagger}+a_V^{\dagger}c_V^{\dagger})^n\sket{0},
\nonumber \\
\sket {\varphi_m} &=\, \frac{1}{m!\,\sqrt{m+1}}(f_H^{\dagger}b_H^{\dagger}+f_V^{\dagger}b_V^{\dagger})^m\sket{0},
\end{align}
where $a_H^{\dagger}$, $c_H^{\dagger}$, $f_H^{\dagger}$ and $b_H^{\dagger}$ ($a_V^{\dagger}$, $c_V^{\dagger}$, $f_V^{\dagger}$ and $b_V^{\dagger}$) are the creation operators of horizontally (vertically) polarized photons of the corresponding modes. 

Next, we calculate the quantum state from Alice that enters the 50:50 beam splitter (BS) within the QND measurement after travelling through the quantum channel. For this, we model the quantum channel by a BS with transmittance $\eta_{\rm{ch}}=\exp(-L/(2L_{\rm{att}}))$. In doing so, it turns out that such a state is given by

\begin{equation}\label{beforeBSFromAlice}
\rho_{ea}= \sum_{n=0}^{\infty}\,p_n \tr_d\left(\ketbra {\phi^n_{aed}} {\phi^n_{aed}}\right),
\end{equation}
where the states $\ket{\phi^n_{aed}}$ are given by

\begin{align}\label{statesBeforeBS}
&\sket{\phi^n_{aed}}= \frac{1}{n!\,\sqrt{n+1}} \sum_{k=0}^{n} \sum_{x=0}^{k} \sum_{y=0}^{n-k} \binom{n}{k}\binom{k}{x}\binom{n-k}{y} \nonumber\\
&\sqrt{1-\eta_{\rm{ch}}}^{x+y}\sqrt{\eta_{\rm{ch}}}^{n-x-y}a_H^{\dagger k} a_V^{\dagger n-k} e_H^{\dagger k-x} e_V^{\dagger n-k-y} \nonumber\\
 &d_H^{\dagger x} d_V^{\dagger y}\sket{0}.
\end{align}
The 50:50 BS combines the states $\rho_{ea}$ and $\rho^{\rm{QND}}$. So, the state after the 50:50 BS can be written as

\begin{equation}\label{stateAfterBS}
\sigma_{abgh}= \sum_{n=0}^{\infty} \sum_{m=0}^{\infty}\,p_n\,q_m \tr_d\left(\ketbra {\varphi^{nm}_{abghd}} {\varphi^{nm}_{abghd}}\right),
\end{equation}
where the pure states $\ket {\varphi^{nm}_{abghd}}$ have the form of

\begin{align}\label{statesAfterBS}
&\sket{\varphi^{nm}_{abghd}}= \frac{1}{n!m!\,\sqrt{(n+1)(m+1)}} \sum_{k=0}^{n} \sum_{x=0}^{k} \sum_{y=0}^{n-k} \sum_{o=0}^{m} \binom{n}{k}\nonumber\\
&\binom{k}{x}\binom{n-k}{y} \binom{m}{o} \frac{\sqrt{1-\eta_{\rm{ch}}}^{x+y}\sqrt{\eta_{\rm{ch}}}^{n-x-y}}{\sqrt{2}^{n+m-x-y}}\sum_{u=0}^{o} \sum_{v=0}^{m-o} \sum_{w=0}^{k-x}\nonumber\\
&\sum_{z=0}^{n-k-y}\binom{k-x}{w}\binom{n-k-y}{z}\binom{o}{u}\binom{m-o}{v}\nonumber\\ 
&(-1)^{n-x-y-w-z}a_H^{\dagger k} a_V^{\dagger n-k} b_H^{\dagger o} b_V^{\dagger m-o}    g_H^{\dagger u+w} g_V^{\dagger v+z} \nonumber\\
&h_H^{\dagger o-u+k-x-w} h_V^{\dagger n-k-y-z+m-o-v} d_H^{\dagger x} d_V^{\dagger y}\sket{0}.
\end{align}

The quantity $p_{\rm{QND}}$ is defined (see Appendix~\ref{App:rephrase}) by the probability of the event that there is exactly one photon detected in each of the modes $g_h$, $g_v$ and $a_h$ (detectors $D_{\rm{2H}}$, $D_{\rm{2V}}$ and $D_{\rm{3H}}$ in Fig.~\ref{Plot:QNDZX}) and zero photons detected in the modes $h_h$, $h_v$ and $a_v$ (detectors $D_{\rm{1H}}$, $D_{\rm{1V}}$ and $D_{\rm{3V}}$ in Fig.~\ref{Plot:QNDZX}). This event is described by the following POVM
\begin{equation}\label{PNRPOVMSuccess}
\Pi_{\rm{QND}}=\Pi_{1}^{g_h}\otimes\Pi_{1}^{g_v}\otimes\Pi_{0}^{h_h}\otimes\Pi_{0}^{h_v}\otimes\Pi_{1}^{a_h}\otimes\Pi_{0}^{a_v},
\end{equation}
where we extended the notation used in Eq.~\eqref{PNRPOVM} with including the corresponding optical mode as a superscript. 

The unnormalised state that enters Charlie's BSM module from Alice (Bob) is then given by
\begin{align}\label{AfterMeasState}
&\Gamma_{b}= \tr_{gha}\left(\Pi_{\rm{QND}}\, \sigma_{abgh}\right)\nonumber \\&=\sum_{n=0}^{\infty} \sum_{m=0}^{\infty}\,p_n\,q_m \tr_{ghad}\left(\Pi_{\rm{QND}}\,\ketbra {\varphi^{nm}_{abghd}} {\varphi^{nm}_{abghd}}\right).
\end{align}
We repeatedly make use of the following transformation of the summation indices, whenever we calculate the trace of an expression.
\begin{align}\label{method}
\sum_{A=0}^{C}\sum_{B=0}^{D}f(A,B)=\sum_{G=0}^{C+D}\sum_{A=\max\lbrace 0,G-D\rbrace}^{\min\lbrace G,C\rbrace}f(A,G-A),
\end{align}
where $f$ is an arbitrary function of the indices $A$, $B$ and we introduced the sum $G=A+B$ of the summation indices.
Carrying out the calculations, it can be shown, by using Eq.~\eqref{method}, that $\Gamma_b$ can be put into the following form
\begin{align}\label{FinalExpression}
&\Gamma_{b}=\sum_{n=1}^{\infty} \sum_{m=0}^{\infty}\,p_n\,q_m \gamma_{b}(n,m),
\end{align}
where $\gamma_{b}(n,m)$ can be written as
\begin{align}\label{gamma}
&\gamma_{b}(n,m)=\sum_{k=1}^{n}\sum_{x=0}^{k}\sum_{y=0}^{n-k}\sum_{o=0}^{m}\sum_{s=1}^{o+k-x}\sum_{t=1}^{m-o+n-k-y}\nonumber \\
&\sum_{u=\max\lbrace 0,s-(k-x) \rbrace}^{\min\lbrace o,s \rbrace}\sum_{v=\max\lbrace 0,t-(n-k-y) \rbrace}^{\min\lbrace m-o,t \rbrace} \sum_{u'=\max\lbrace 0,s-(k-x) \rbrace}^{\min\lbrace o,s \rbrace}\nonumber \\
&\sum_{v'=\max\lbrace 0,t-(n-k-y) \rbrace}^{\min\lbrace m-o,t \rbrace} \Lambda(n,m,k,x,y,o,s,t,u,v,u',v')\nonumber \\
&\ketbraMode{o,m-o}{o,m-o}{b},
\end{align}
where $\ket{i,j}_b$ denotes that there are $i$ and $j$ photons in modes $b_h$ and $b_v$, respectively. The quantity $\Lambda(n,m,k,x,y,o,s,t,u,v,u',v')$, on the other hand, equals to 

\begin{widetext}
\begin{align}\label{Lambda}
&\Lambda(n,m,k,x,y,o,s,t,u,v,u',v')=\frac{s\,t\,k\,\eta_{\rm{det}}^3\,(1-\eta_{\rm{det}})^{2n+m-x-y-3}\, k!\,(n-k)!\,o!\,(m-o)!\,s!\,t!\,(o+k-x-s)!}{(n+1)\,(m+1)\,x!\,y!\,(k-x+u-s)!\,(k-x+u'-s)!\,(s-u)!\,(s-u')}\nonumber \\
&\frac{(n-k-y+m-o-t)!\,(-1)^{u+v+u'+v'}}{(n-k-y+v-t)!\,(t-v)!\,(n-k-y+v'-t)!\,(t-v')!\,(o-u)!\,u!\,(o-u')!\,u'!\,(m-o-v)!\,v!\,(m-o-v')!\,v'!}\nonumber \\
&\frac{1}{2^{n+m-x-y}}\,\eta_{\rm{ch}}^{n-x-y}\,(1-\eta_{\rm{ch}})^{x+y}.
\end{align}
\end{widetext}

The probability $p_{\rm{QND}}$ can be obtained as the normalization factor of $\Gamma_b$:

\begin{align}\label{probabilityQNDZ}
&p_{\rm{QND}}=\tr[\Gamma_b]=\sum_{n=1}^{\infty} \sum_{m=0}^{\infty}\,p_n\,q_m\sum_{k=1}^{n}\sum_{x=0}^{k}\sum_{y=0}^{n-k}\sum_{o=0}^{m}\nonumber \\
&\sum_{s=1}^{o+k-x}\sum_{t=1}^{m-o+n-k-y}\sum_{u=\max\lbrace 0,s-(k-x) \rbrace}^{\min\lbrace o,s \rbrace}\sum_{v=\max\lbrace 0,t-(n-k-y) \rbrace}^{\min\lbrace m-o,t \rbrace}\nonumber \\
& \sum_{u'=\max\lbrace 0,s-(k-x) \rbrace}^{\min\lbrace o,s \rbrace}\sum_{v'=\max\lbrace 0,t-(n-k-y) \rbrace}^{\min\lbrace m-o,t \rbrace} \nonumber \\
&\Lambda(n,m,k,x,y,o,s,t,u,v,u',v').
\end{align}
\subsection{Derivation of $p_{\rm{c}}^{\rm{Z}}$ and $p_{\rm{nc}}^{\rm{Z}}$}\label{App:pcAndpnc}
The layout for this derivation can be seen in Fig.~\ref{Plot:BSM}. Using the previous result for the state coming from the QND and the $Z$ measurement, which is given by Eq.~\eqref{FinalExpression}, the state that enters Charlie's BSM module from Alice's (Bob's) side is $\Gamma_{c}$ ($\Gamma_{d}$). Therefore, their collective state can be written as
\begin{widetext}
\begin{align}\label{InputBSM}
&\Gamma_{c}\otimes\Gamma_{d}= \sum_{n=1}^{\infty} \sum_{m=0}^{\infty}\sum_{N=1}^{\infty} \sum_{M=0}^{\infty}\,p_n\,q_m\,p_N\,q_M \gamma_{c}(n,m)\otimes\gamma_{d}(N,M)=  \sum_{n=1}^{\infty} \sum_{m=0}^{\infty}\sum_{k=1}^{n}\sum_{x=0}^{k}\sum_{y=0}^{n-k}\sum_{o=0}^{m}\sum_{s=1}^{o+k-x}\sum_{t=1}^{m-o+n-k-y}\nonumber \\
&\sum_{u=\max\lbrace 0,s-(k-x) \rbrace}^{\min\lbrace o,s \rbrace}\sum_{v=\max\lbrace 0,t-(n-k-y) \rbrace}^{\min\lbrace m-o,t \rbrace} \sum_{u'=\max\lbrace 0,s-(k-x) \rbrace}^{\min\lbrace o,s \rbrace}\sum_{v'=\max\lbrace 0,t-(n-k-y) \rbrace}^{\min\lbrace m-o,t \rbrace}\sum_{N=1}^{\infty} \sum_{M=0}^{\infty}\sum_{K=1}^{N}\sum_{X=0}^{K}\sum_{Y=0}^{N-K}\sum_{O=0}^{M}\sum_{S=1}^{O+K-X}\nonumber \\
&\sum_{T=1}^{M-O+N-K-Y}\sum_{U=\max\lbrace 0,S-(K-X) \rbrace}^{\min\lbrace O,S \rbrace}\sum_{V=\max\lbrace 0,T-(N-K-Y) \rbrace}^{\min\lbrace M-O,T \rbrace} \sum_{U'=\max\lbrace 0,S-(K-X) \rbrace}^{\min\lbrace O,S \rbrace}\sum_{V'=\max\lbrace 0,T-(N-K-Y) \rbrace}^{\min\lbrace M-O,T \rbrace}\,p_n\,q_m\,p_N\,q_M\,\Lambda_c\,\Lambda_d\nonumber \\
&\ketbraMode{o,m-o,O,M-O}{o,m-o,O,M-O}{cd}=\nonumber \\
&=\,\sum_{n,\cdots ,V'} \,p_n\,q_m\,p_N\,q_M\,\Lambda_c\,\Lambda_d\,\ketbraMode{o,m-o,O,M-O}{o,m-o,O,M-O}{cd},
\end{align}
\end{widetext}
where note that upper case indices are used to describe quantities corresponding to mode $d$ (coming from Bob's side), while using lower case indices to describe quantities corresponding to mode $c$ (coming from Alice's side). For simplicity, in Eq.~\eqref{InputBSM} we introduced the following notation:
\begin{align}\label{notationLambda}
&\Lambda_c\equiv\Lambda(n,m,k,x,y,o,s,t,u,v,u',v'),\nonumber \\
&\Lambda_d\equiv\Lambda(N,M,K,X,Y,O,S,T,U,V,U',V').
\end{align}
For brevity, after the last equation sign in Eq.~\eqref{InputBSM} we denoted all the sums  collectively by $\sum_{n,\cdots ,V'}$. The state after the 50:50 BS can then be written concisely as
\begin{align}\label{AfterBSInBSM}
\rho_{ef}=\sum_{n,\cdots ,V'} p_n\,q_m\,p_N\,q_M\,\Lambda_c\,\Lambda_d\ketbra{\psi^{omOM}_{ef}}{\psi^{omOM}_{ef}},
\end{align}
where the states $\ket{\psi^{omOM}_{ef}}$ are given by
\begin{align}\label{AfterBSInBSMPureStates}
&\ket{\psi^{omOM}_{ef}}=\frac{1}{\sqrt{o!\,(m-o)!\,O!\,(M-O)!}}\left(\frac{1}{\sqrt{2}}\right)^{m+M}\nonumber \\
&\sum_{l=0}^{o}\sum_{q=0}^{m-o}\sum_{L=0}^{O}\sum_{Q=0}^{M-O}\binom o l\binom O L \binom{m-o}{q}\binom{M-O}{Q}\nonumber \\
&\, (-1)^{m-l-q} e_H^{\dagger o+O-l-L} e_V^{\dagger m+M-o-O-q-Q} f_H^{\dagger l+L} f_V^{\dagger q+Q}\ket{0}.
\end{align}

Then, the state after the Hadamard gates can be written as
\begin{align}\label{AfterHInBSM}
\rho_{gh}=\sum_{n,\cdots ,V'} p_n\,q_m\,p_N\,q_M\,\Lambda_c\,\Lambda_d\ketbra{\chi^{omOM}_{gh}}{\chi^{omOM}_{gh}},
\end{align}
where the pure states $\ket{\chi^{omOM}_{gh}}$ have the form 
\begin{widetext}
\begin{align}\label{AfterHInBSMPureStates}
&\ket{\chi^{omOM}_{gh}}=\sum_{l=0}^{o}\sum_{q=0}^{m-o}\sum_{L=0}^{O}\sum_{Q=0}^{M-O} \sum_{\alpha=0}^{l+L}\sum_{\beta=0}^{q+Q}\sum_{\varphi=0}^{o+O-l-L}\sum_{\varepsilon=0}^{m+M-o-O-q-Q} \frac{1}{\sqrt{o!\,(m-o)!\,O!\,(M-O)!}}\left(\frac{1}{2}\right)^{m+M}\nonumber \\
&(-1)^{M-l-q-o-O-\varepsilon-\beta}\binom o l\binom O L \binom{m-o}{q}\binom{M-O}{Q}\binom{l+L}{\alpha}\binom{q+Q}{\beta}\binom{o+O-l-L}{\varphi}\binom{m+M-o-O-q-Q}{\varepsilon}\nonumber\\
& g_H^{\dagger \alpha+\beta}h_H^{\dagger \varepsilon+\varphi}\, g_V^{\dagger l+L+q+Q-\alpha-\beta}h_V^{\dagger m+M-l-L-q-Q-\varepsilon-\varphi}\ket{0}.
\end{align}
\end{widetext}

With Eq.~\eqref{method}, we can rewrite the states $\ket{\chi^{omOM}_{gh}}$ into a form, in which it will be more convenient to take the trace of Eq.~\eqref{AfterHInBSM}:
\begin{widetext}
\begin{align}\label{AfterHInBSMPureStatesVer2}
&\ket{\chi^{omOM}_{gh}}=\sum_{I=0}^{m+M}\sum_{\Omega =0}^{m+M}\sum_{i=\max\lbrace 0,I+o+O-m-M\rbrace}^{\min\lbrace I,o+O\rbrace} \nonumber\sum_{\omega=\max\lbrace 0,\Omega+I-m-M\rbrace}^{\min\lbrace \Omega,I\rbrace} \sum_{l=\max\lbrace 0,i-O\rbrace}^{\min\lbrace i,o\rbrace}\sum_{q=\max\lbrace 0,I+O-i-M\rbrace}^{\min\lbrace I-i,m-o\rbrace}\sum_{\alpha=\max\lbrace 0,\omega+i-I\rbrace}^{\min\lbrace \omega,i\rbrace}\nonumber\\
&\sum_{\varphi=\max\lbrace 0,\Omega+O+I+o-\omega-i-m-M\rbrace}^{\min\lbrace \Omega-\omega,o+O-i\rbrace}\ \sqrt{\frac{\omega!(I-\omega)!(\Omega-\omega)!(m+M+\omega-I-\Omega)!}{o!\,(m-o)!\,O!\,(M-O)!}}\left(\frac{1}{2}\right)^{m+M}(-1)^{M+\varphi+\alpha-l-q-o-O-\Omega}\binom o l\nonumber\\
&\binom O {i-l} \binom{m-o}{q}\binom{M-O}{I-i-q}\binom{i}{\alpha}\binom{I-i}{\omega-\alpha}\binom{o+O-i}{\varphi}\binom{m+M+i-o-O-I}{\Omega-\omega-\varphi}  \nonumber\\
&\ket{\omega,I-\omega,\Omega-\omega, m+M+\omega-I-\Omega}_{gh}.
\end{align}
\end{widetext}
The quantity $p_{\rm{c}}^{\rm{Z}}$ is defined (see Appendix~\ref{App:rephrase}) by the probability of the event that, one photon is observed in each of the modes $g_h$ and $g_v$ (one detection in each of the detectors $D_{\rm{2H}}$ and $D_{\rm{2V}}$ in Fig.~\ref{Plot:BSM}) and zero photons are observed in each of the modes $h_h$ and $h_v$ (no detection in neither of the detectors $D_{\rm{1H}}$ and $D_{\rm{1V}}$ in Fig.~\ref{Plot:BSM}) 

Similarly, $p_{\rm{nc}}^{\rm{Z}}$ is by definition equal to the probability of the event that one photon is observed in each of the modes $g_h$ and $h_v$ (one detection in each of the detectors $D_{\rm{2H}}$ and $D_{\rm{1V}}$ in Fig.~\ref{Plot:BSM}) and zero photons are observed in each of the modes $g_v$ and $h_h$ (no detection in neither of the detectors $D_{\rm{1H}}$ and $D_{\rm{2V}}$ in Fig.~\ref{Plot:BSM}).

These events are described by the following POVMs
\begin{equation}\label{pcPOVM}
\Pi_{c}=\Pi_{1}^{g_h}\otimes\Pi_{1}^{g_v}\otimes\Pi_{0}^{h_h}\otimes\Pi_{0}^{h_v},
\end{equation}
and
\begin{equation}\label{pncPOVM}
\Pi_{nc}=\Pi_{1}^{g_h}\otimes\Pi_{0}^{g_v}\otimes\Pi_{0}^{h_h}\otimes\Pi_{1}^{h_v}.
\end{equation}
Since for convenience in the calculations we incorporate the loss corresponding to the active feedforward mechanism into the efficiency of the detectors in the BSM module, the efficiency becomes
\begin{equation}\label{detEffBSM}
\eta'_{\rm{det}}=\eta_{\rm{det}}\eta_{\rm{f}},\,\text{with}\quad \eta_{\rm{f}}=\exp(-\,c\,\tau/L_{\rm{att}}).
\end{equation} 
This is the efficiency assumed in the POVM elements given by Eq.~\eqref{pcPOVM} and Eq.~\eqref{pncPOVM}, where $\tau$ is the necessary time for performing one active feedforward and $c$ is the speed of light in the optical fiber. 

Therefore, $p_{\rm{c}}^{\rm{Z}}$ and $p_{\rm{nc}}^{\rm{Z}}$ can be calculated as follows
\begin{equation}\label{pc}
p_{\rm{c}}^{\rm{Z}}= \tr_{gh}\left(\Pi_{\rm{c}}\, \rho_{gh}\right),
\end{equation}
and
\begin{equation}\label{pnc}
p_{\rm{nc}}^{\rm{Z}}= \tr_{gh}\left(\Pi_{\rm{nc}}\, \rho_{gh}\right),
\end{equation}
with $\rho_{gh}$ given by Eq.~\eqref{AfterHInBSM}.
Carrying out the calculations and plugging all the indices back in, we find that $p_{\rm{c}}^{\rm{Z}}$ and $p_{\rm{nc}}^{\rm{Z}}$ are given by the following formulas:
\begin{widetext}
\begin{align}\label{pcFinal}
&p_{\rm{c}}^{\rm{Z}}=\sum_{n=1}^{\infty} \sum_{m=0}^{\infty}\sum_{N=1}^{\infty}\sum_{M=0}^{\infty}p_n\,q_m\,p_N\,q_M\sum_{k=1}^{n}\sum_{x=0}^{k}\sum_{y=0}^{n-k}\sum_{o=0}^{m}\sum_{s=1}^{o+k-x}\sum_{t=1}^{m-o+n-k-y}\sum_{u=\max\lbrace 0,s+x-k \rbrace}^{\min\lbrace o,s \rbrace}\sum_{v=\max\lbrace 0,t+k+y-n \rbrace}^{\min\lbrace m-o,t \rbrace} \sum_{u'=\max\lbrace 0,s+x-k \rbrace}^{\min\lbrace o,s \rbrace}\nonumber\\
& \sum_{v'=\max\lbrace 0,t+k+y-n \rbrace}^{\min\lbrace m-o,t \rbrace}\sum_{K=1}^{N}\sum_{X=0}^{K}\sum_{Y=0}^{N-K}\sum_{O=0}^{M}\sum_{S=1}^{O+K-X}\sum_{T=1}^{M-O+N-K-Y}\sum_{U=\max\lbrace 0,S+X-K \rbrace}^{\min\lbrace O,S \rbrace}\sum_{V=\max\lbrace 0,T+K+Y-N \rbrace}^{\min\lbrace M-O,T \rbrace} \sum_{U'=\max\lbrace 0,S+X-K \rbrace}^{\min\lbrace O,S \rbrace}\nonumber \\
&\sum_{V'=\max\lbrace 0,T+K+Y-N \rbrace}^{\min\lbrace M-O,T \rbrace}\,\Lambda_c\,\Lambda_d\sum_{I=1}^{m+M}\sum_{\Omega=1}^{m+M}\sum_{i=\max\lbrace 0,I+o+O-m-M\rbrace}^{\min\lbrace I,o+O\rbrace} \nonumber\sum_{\omega=\max\lbrace 1,\Omega+I-m-M\rbrace}^{\min\lbrace \Omega,I-1\rbrace} \sum_{l=\max\lbrace 0,i-O\rbrace}^{\min\lbrace i,o\rbrace}\sum_{q=\max\lbrace 0,I+O-i-M\rbrace}^{\min\lbrace I-i,m-o\rbrace}\nonumber \\
&\sum_{\alpha=\max\lbrace 0,\omega+i-I\rbrace}^{\min\lbrace \omega,i\rbrace}\sum_{\varphi=\max\lbrace 0,\Omega+O+I+o-\omega-i-m-M\rbrace}^{\min\lbrace \Omega-\omega,o+O-i\rbrace}   \sum_{i'=\max\lbrace 0,I+o+O-m-M\rbrace}^{\min\lbrace I,o+O\rbrace} \nonumber\sum_{l'=\max\lbrace 0,i'-O\rbrace}^{\min\lbrace i',o\rbrace}\sum_{q'=\max\lbrace 0,I+O-i'-M\rbrace}^{\min\lbrace I-i',m-o\rbrace}\nonumber \\
&\sum_{\alpha'=\max\lbrace 0,\omega+i'-I\rbrace}^{\min\lbrace \omega,i'\rbrace}\sum_{\varphi'=\max\lbrace 0,\Omega+O+I+o-\omega-i'-m-M\rbrace}^{\min\lbrace \Omega-\omega,o+O-i'\rbrace}\,\omega\,(I-\omega)\,(\eta'_{\rm{det}})^2 \,(1-\eta'_{\rm{det}})^{m+M-2}\nonumber \\
&G(m,o,M,O,I,\Omega,i,\omega,l,q,\alpha, \varphi,i',l',q',\alpha', \varphi' ),
\end{align}
\end{widetext}
and
\begin{widetext}
\begin{align}\label{pncFinal}
&p_{\rm{nc}}^{\rm{Z}}=\sum_{n=1}^{\infty} \sum_{m=0}^{\infty}\sum_{N=1}^{\infty}\sum_{M=0}^{\infty}p_n\,q_m\,p_N\,q_M\sum_{k=1}^{n}\sum_{x=0}^{k}\sum_{y=0}^{n-k}\sum_{o=0}^{m}\sum_{s=1}^{o+k-x}\sum_{t=1}^{m-o+n-k-y}\sum_{u=\max\lbrace 0,s+x-k \rbrace}^{\min\lbrace o,s \rbrace}\sum_{v=\max\lbrace 0,t+k+y-n \rbrace}^{\min\lbrace m-o,t \rbrace} \sum_{u'=\max\lbrace 0,s+x-k \rbrace}^{\min\lbrace o,s \rbrace}\nonumber\\
& \sum_{v'=\max\lbrace 0,t+k+y-n \rbrace}^{\min\lbrace m-o,t \rbrace}\sum_{K=1}^{N}\sum_{X=0}^{K}\sum_{Y=0}^{N-K}\sum_{O=0}^{M}\sum_{S=1}^{O+K-X}\sum_{T=1}^{M-O+N-K-Y}\sum_{U=\max\lbrace 0,S+X-K \rbrace}^{\min\lbrace O,S \rbrace}\sum_{V=\max\lbrace 0,T+K+Y-N \rbrace}^{\min\lbrace M-O,T \rbrace} \sum_{U'=\max\lbrace 0,S+X-K \rbrace}^{\min\lbrace O,S \rbrace}\nonumber \\
&\sum_{V'=\max\lbrace 0,T+K+Y-N \rbrace}^{\min\lbrace M-O,T \rbrace}\,\Lambda_c\,\Lambda_d\sum_{I=1}^{m+M}\sum_{\Omega=1}^{m+M}\sum_{i=\max\lbrace 0,I+o+O-m-M\rbrace}^{\min\lbrace I,o+O\rbrace}\sum_{\omega=\max\lbrace 1,1+\Omega+I-m-M \rbrace}^{\min\lbrace \Omega,I\rbrace} \sum_{l=\max\lbrace 0,i-O\rbrace}^{\min\lbrace i,o\rbrace}\sum_{q=\max\lbrace 0,I+O-i-M\rbrace}^{\min\lbrace I-i,m-o\rbrace}\nonumber \\
&\sum_{\alpha=\max\lbrace 0,\omega+i-I\rbrace}^{\min\lbrace \omega,i\rbrace}\sum_{\varphi=\max\lbrace 0,\Omega+O+I+o-\omega-i-m-M\rbrace}^{\min\lbrace \Omega-\omega,o+O-i\rbrace}   \sum_{i'=\max\lbrace 0,I+o+O-m-M\rbrace}^{\min\lbrace I,o+O\rbrace} \nonumber\sum_{l'=\max\lbrace 0,i'-O\rbrace}^{\min\lbrace i',o\rbrace}\sum_{q'=\max\lbrace 0,I+O-i'-M\rbrace}^{\min\lbrace I-i',m-o\rbrace}\nonumber \\
&\sum_{\alpha'=\max\lbrace 0,\omega+i'-I\rbrace}^{\min\lbrace \omega,i'\rbrace}\sum_{\varphi'=\max\lbrace 0,\Omega+O+I+o-\omega-i'-m-M\rbrace}^{\min\lbrace \Omega-\omega,o+O-i'\rbrace}\,\omega\,(m+M+\omega-I-\Omega)\,(\eta'_{\rm{det}})^2 \,(1-\eta'_{\rm{det}})^{m+M-2}\,\nonumber \\
&G(m,o,M,O,I,\Omega,i,\omega,l,q,\alpha, \varphi,i',l',q',\alpha', \varphi' ).
\end{align}
\end{widetext}  

Note, that the only differences between $p_{\rm{c}}^{\rm{Z}}$ and $p_{\rm{nc}}^{\rm{Z}}$ are in the limits of the index $\omega$ and in the expression after the sums. Moreover, the term $G(m,o,M,O,I,\Omega,i,\omega,l,q,\alpha, \varphi,i',l',q',\alpha', \varphi' )$ is given by the following expression

\begin{widetext}
\begin{align}\label{G}
&G(m,o,M,O,I,\Omega,i,\omega,l,q,\alpha, \varphi,i',l',q',\alpha', \varphi')=\frac{\omega!(I-\omega)!(\Omega-\omega)!(m+M+\omega-I-\Omega)!\,(-1)^{\varphi+\alpha+\varphi'+\alpha'-l-q-l'-q'}}{o!\,(m-o)!\,O!\,(M-O)!\,4^{m+M}}\nonumber\\
&\binom o l\binom O {i-l} \binom{m-o}{q}\binom{M-O}{I-i-q}\binom{i}{\alpha}\binom{I-i}{\omega-\alpha}\binom{o+O-i}{\varphi}\binom{m+M+i-o-O-I}{\Omega-\omega-\varphi} \binom o {l'}\binom O {i'-l'} \binom{m-o}{q'}\nonumber\\
&\binom{M-O}{I-i'-q'}\binom{i'}{\alpha'}\binom{I-i'}{\omega-\alpha'}\binom{o+O-i'}{\varphi'}\binom{m+M+i'-o-O-I}{\Omega-\omega-\varphi'}.
\end{align}
\end{widetext} 

\subsection{Derivation of $p_{\rm{c}}^{\rm{X}}$ and $p_{\rm{nc}}^{\rm{X}}$}\label{App:pcAndpncX}

The derivation of the probabilities $p_{\rm{c}}^{\rm{X}}$ and $p_{\rm{nc}}^{\rm{X}}$ is very similar to the derivation of $p_{\rm{c}}^{\rm{Z}}$ and $p_{\rm{nc}}^{\rm{Z}}$ in Appendix~\ref{App:pcAndpnc}. However, for the $X$-basis we perform the derivation with a slightly different structure, that is, we perform the $X$ measurements at the very end, after Charlie's QND measurement and BSM have gone through successfully. However, we remark that the calculations could also be done using the same structure like in Appendix~\ref{App:pcAndpnc}.

So firstly, we obtain the state that Alice (Bob) has after Charlie's QND measurement was performed successfully on her (his) side, which is described by the following POVM:
\begin{equation}\label{QNDPOVMX}
\Pi=\Pi_{1}^{g_h}\otimes\Pi_{1}^{g_v}\otimes\Pi_{0}^{h_h}\otimes\Pi_{0}^{h_v},
\end{equation}
where note that we use the same notation for the modes as in Fig.~\ref{Plot:QNDZX} and we excluded the $X$ measurement so far.  Let us denote this state by $\sigma_{AC}$ ($\sigma_{BC}$) on Alice's (Bob's) side. Then, we take the tensor product $\sigma_{AC}\otimes\sigma_{BC}$ and perform the middle BSM with the POVM of Eq.~\eqref{QNDPOVMX}, but here the modes correspond to Fig.~\ref{Plot:BSM} since this is the measurement being executed. Next, on the obtained state from the BSM, the parties perform the $X$ measurement, which essentially means that they apply Hadamard gates on their modes and after that they perform the $Z$ measurement in Fig.~\ref{Plot:QNDZX}. Let $\sigma'_{ab}$ denote the state held by Alice and Bob after they apply the Hadamard gates, where $a$ ($b$) represents the optical mode entering the PNR detectors in the $Z$ measurement at Alice's (Bob's) site. 

Considering the detection patterns that define the correlated (correct) and anti-correlated (non-correct) raw key generation events, as explained previously in Appendix~\ref{App:rephrase}, we have that
\begin{equation}\label{pcX}
p_{\rm{c}}^{\rm{X}}= \tr_{ab}\left(\Pi^{\rm{X}}_{\rm{c}}\, \sigma'_{ab}\right),
\end{equation}

\begin{equation}\label{pncX}
p_{\rm{nc}}^{\rm{X}}= \tr_{ab}\left(\Pi^{\rm{X}}_{\rm{nc}}\, \sigma'_{ab}\right),
\end{equation}
with
\begin{equation}\label{POVMXc}
\Pi^{\rm{X}}_{\rm{c}}=\Pi_{1}^{a_h}\otimes\Pi_{0}^{a_v}\otimes\Pi_{0}^{b_h}\otimes\Pi_{1}^{b_v},
\end{equation}
and
\begin{equation}\label{POVMXnc}
\Pi^{\rm{X}}_{\rm{nc}}=\Pi_{1}^{a_h}\otimes\Pi_{0}^{a_v}\otimes\Pi_{1}^{b_h}\otimes\Pi_{0}^{b_v}.
\end{equation}

Here, for simplicity we only present the results for the main steps of the derivation. The state $\sigma_{AC}$ is given by the following formula:
\begin{align}\label{QNDX}
&\sigma_{AC}=\sum_{n=0}^{\infty} \sum_{m=0}^{\infty}\,p_n\,q_m \sum_{k=0}^{n}\sum_{x=0}^{k}\sum_{y=0}^{n-k}\sum_{o=0}^{m}\sum_{k'=\max \lbrace x,o+k-m \rbrace}^{\min \lbrace n-y,o+k \rbrace}\nonumber \\
& \Lambda_{\rm{QND}}(n,m,k,x,y,o,k')\,\ketMode{k,n-k,o,m-o}{a'c}\nonumber \\
&\braMode{k',n-k',o+k-k',m-o-k+k'}{a'c},
\end{align}
where $a'$ ($c$) is the mode that enters Alice's $X$ measurement (Charlie's BSM from Alice's side) and $\Lambda_{\rm{QND}}(n,m,k,x,y,o,k')$ is given by the following expression:

\begin{widetext}
\begin{align}\label{LambdaQND}
&\Lambda_{\rm{QND}}(n,m,k,x,y,o,k')=\frac{\eta_{\rm{det}}^2\,(1-\eta_{\rm{det}})^{n+m-x-y-2}\,\eta_{\rm{ch}}^{n-x-y}\,(1-\eta_{\rm{ch}})^{x+y} \sqrt{k!\,(n-k)!(k')!\,(n-k')!\,o!\,(m-o)!}}{(n+1)\,(m+1)\,x!\,y!\,2^{n+m-x-y}}\nonumber \\
&\sqrt{(o+k-k')!(m-o-k+k')!}\sum_{u=0}^{o}\sum_{v=0}^{m-o}\sum_{w=\max\lbrace 0,1-u \rbrace}^{k-x}\sum_{z=\max\lbrace 0,1-v \rbrace}^{n-k-y} \sum_{u'=\max\lbrace 0,u+w+x-k' \rbrace}^{\min\lbrace u+w,o+k-k' \rbrace}\sum_{v'=\max\lbrace 0,v+z+y+k'-n \rbrace}^{\min\lbrace v+z,m-o-k+k' \rbrace}\nonumber \\
&\frac{(-1)^{u'+v'-u-v}(u+w)(v+z)(u+w)!(v+z)!(o-u+k-x-w)!(n-k-y-z+m-o-v)!}{u!\,v!\,w!\,z!\,(u')!\,(v')!(k-x-w)!(n-k-y-z)!(o-u)!(m-o-v)!(u+w-u')!(k'-x-u-w+u')!(v+z-v')!}\nonumber \\
&\frac{1}{(n-k'-y-v-z+v')!\,(o+k-k'-u')!\,(m-o-k+k'-v')!}.
\end{align}
\end{widetext}
We note that $\sigma_{BC}$ has the exact same form with mode $b'$ ($d$) entering Bob's $X$ measurement (Charlie's BSM from Bob's side) and in the expression of $\sigma_{BC}$ we use capital letter indices similarly to Appendix~\ref{App:pcAndpnc}:
\begin{align}\label{QNDXBC}
&\sigma_{BC}=\sum_{N=0}^{\infty} \sum_{M=0}^{\infty}\,p_N\,q_M \sum_{K=0}^{N}\sum_{X=0}^{K}\sum_{Y=0}^{N-K}\sum_{O=0}^{M}\nonumber \\
&\sum_{K'=\max \lbrace X,O+K-M \rbrace}^{\min \lbrace N-Y,O+K \rbrace} \Lambda_{\rm{QND}}(N,M,K,X,Y,O,K')\nonumber \\
&\ketMode{K,N-K,O,M-O}{b'd}\nonumber \\
&\braMode{K',N-K',O+K-K',M-O-K+K'}{b'd}.
\end{align}
Then we obtain $\sigma'_{ab}$ by performing Charlie's BSM on modes $c$ and $d$ and applying Hadamard gates to the modes $a'$ and $b'$:
\begin{widetext}
\begin{align}\label{FinalStateX}
&\sigma'_{ab}=\sum_{n=0}^{\infty} \sum_{m=0}^{\infty}\sum_{N=0}^{\infty} \sum_{M=0}^{\infty}\,p_n\,q_m\,p_N\,q_M \sum_{k=0}^{n}\sum_{x=0}^{k}\sum_{y=0}^{n-k}\sum_{o=0}^{m}\sum_{k'=\max \lbrace x,o+k-m \rbrace}^{\min \lbrace n-y,o+k \rbrace}\sum_{K=0}^{N}\sum_{X=0}^{K}\sum_{Y=0}^{N-K}\sum_{O=0}^{M}\sum_{K'=\max \lbrace X,O+K-M \rbrace}^{\min \lbrace N-Y,O+K \rbrace}\nonumber \\
&\sum_{l=0}^{o}\sum_{q=0}^{m-o}\sum_{L=0}^{O}\sum_{Q=0}^{M-O}\sum_{l'=0}^{o+k-k'}\sum_{q'=0}^{m-o-k+k'}\sum_{L'=\max \lbrace 0,q+Q+l+L-q'-l'-M+O+K-K' \rbrace}^{\min \lbrace O+K-K',q+Q+l+L-q'-l'\rbrace}\sum_{\alpha'=0}^{l'+L'}\sum_{\beta'=\max \lbrace 0,1-\alpha' \rbrace}^{q+Q+l+L-l'-L'-\max \lbrace 0,1+\alpha'-l'-L'\rbrace}\nonumber \\
&\sum_{\alpha=\max \lbrace 0,\alpha'+\beta'-q-Q\rbrace}^{\min \lbrace l+L,\alpha'+\beta'\rbrace}\sum_{\varphi=0}^{o-l+O-L}\sum_{\varepsilon=0}^{m-o-q+M-O-Q}\sum_{\varphi'=\max \lbrace 0,\varphi+\varepsilon-m+o+k-k'-M+O+K-K'+q+Q+l+L-l'-L'\rbrace}^{\min \lbrace \varphi+\varepsilon,o+k-k'-l'+O+K-K'-L'\rbrace}G_{\rm{X}}\nonumber \\
&\sum_{\tau=0}^{k}\sum_{\nu=0}^{n-k}\sum_{\chi=0}^{K}\sum_{\omega=0}^{N-K}\sum_{\tau'=0}^{k'}\sum_{\nu'=0}^{n-k'}\sum_{\chi'=0}^{K'}\sum_{\omega'=0}^{N-K'}\frac{f_{\rm{H}}(n,N,k,K,\tau,\nu,\chi,\omega)f_{\rm{H}}(n,N,k',K',\tau',\nu',\chi',\omega')}{g(n,N,\tau+\nu,\omega+\chi,\tau'+\nu',\omega'+\chi')}\nonumber \\
&\ketbraMode{\tau+\nu,n-\tau-\nu,\omega+\chi,N-\omega-\chi}{\tau'+\nu',n-\tau'-\nu',\omega'+\chi',N-\omega'-\chi'}{ab},
\end{align}
\end{widetext}
where, for the sake of convenience, we introduced the following functions:
\begin{widetext}
\begin{align}\label{GX}
&G_{\rm{X}}\equiv G_{\rm{X}}(n,m,N,M,k,x,y,o,k',K,X,Y,O,K',l,q,L,Q,l',q',L',\alpha',\beta',\alpha,\varphi,\varepsilon,\varphi')=\Lambda_{\rm{QND}}(n,m,k,x,y,o,k')\nonumber \\
&\Lambda_{\rm{QND}}(N,M,K,X,Y,O,K')\,(\alpha'+\beta')!(\varphi+\varepsilon)!(l+L+q+Q-\alpha'-\beta')!\frac{(-1)^{L-L'-q'-q+\alpha-\alpha'-\varphi'-\varphi}}{2^{m+M}}\nonumber \\
&(m+M-l-L-q-Q-\varphi-\varepsilon)!\,f_{\rm{H}}(m,M,o+k-k',O+K-K',l',q',L',q+Q+l+L-q'-l'-L')\nonumber \\
&f_{\rm{H}}(m,M,o,O,l,q,L,Q)\binom{l+L}{\alpha}\binom{q+Q}{\alpha'+\beta'-\alpha}\binom{o-l+O-L}{\varphi} \binom{m-o-q+M-O-Q}{\varepsilon}\binom{l'+L'}{\alpha'}\nonumber \\
&\binom{q+Q+l+L-l'-L'}{\beta'}\binom{m-o-k+k'+M-O-K+K'-q-Q-l-L+l'+L'}{\varphi+\varepsilon-\varphi'}\nonumber \\
&\binom{o+k-k'-l'+O+K-K'-L'}{\varphi'}g(m,M,o,O,o+k-k',O+K-K')\,g(n,N,k,K,k',K')\nonumber \\
&d(m,M,\alpha'+\beta',q+Q+l+L-\alpha'-\beta',\eta'_{\rm{det}}),
\end{align}
\end{widetext}
and
\begin{align}\label{g}
&g(n,N,k,K,k',K')=\frac{1}{\sqrt{k!(n-k)!K!(N-K)!}}\nonumber \\
&\frac{1}{\sqrt{(k')!(n-k')!(K')!(N-K')!}},
\end{align} 
which enters into the expressions when the Fock-states are expressed with the creation/annihilation operators, and
\begin{align}\label{fH}
&f_{\rm{H}}(n,N,k,K,\tau,\nu,\chi,\omega)=\frac{(-1)^{-k-K-\nu-\omega}}{\sqrt{2}^{n+N}}\binom{k}{\tau}\binom{n-k}{\nu}\nonumber \\
&\binom{K}{\chi}\binom{N-K}{\omega},
\end{align}
which is introduced due to the Hadamard gates included in the implementation, and
\begin{align}\label{d}
&d(m,M,x,y,\eta)=x\,y\,\eta^2(1-\eta)^{m+M-2},
\end{align}
which comes from the successful detection pattern's POVM. 
With the formula for $\sigma'_{ab}$, given by Eq.~\eqref{FinalStateX}, and Eqs.~\eqref{pcX}-\eqref{pncX} we can obtain
\begin{widetext}
\begin{align}\label{pcFinalX}
&p_{\rm{c}}^{\rm{X}}=\sum_{n=0}^{\infty} \sum_{m=0}^{\infty}\sum_{N=0}^{\infty} \sum_{M=0}^{\infty}\,p_n\,q_m\,p_N\,q_M \sum_{k=0}^{n}\sum_{x=0}^{k}\sum_{y=0}^{n-k}\sum_{o=0}^{m}\sum_{k'=\max \lbrace x,o+k-m \rbrace}^{\min \lbrace n-y,o+k \rbrace}\sum_{K=0}^{N}\sum_{X=0}^{K}\sum_{Y=0}^{N-K}\sum_{O=0}^{M}\sum_{K'=\max \lbrace X,O+K-M \rbrace}^{\min \lbrace N-Y,O+K \rbrace}\nonumber \\
&\sum_{l=0}^{o}\sum_{q=0}^{m-o}\sum_{L=0}^{O}\sum_{Q=0}^{M-O}\sum_{l'=0}^{o+k-k'}\sum_{q'=0}^{m-o-k+k'}\sum_{L'=\max \lbrace 0,q+Q+l+L-q'-l'-M+O+K-K' \rbrace}^{\min \lbrace O+K-K',q+Q+l+L-q'-l'\rbrace}\sum_{\alpha'=0}^{l'+L'}\sum_{\beta'=\max \lbrace 0,1-\alpha' \rbrace}^{q+Q+l+L-l'-L'-\max \lbrace 0,1+\alpha'-l'-L'\rbrace}\nonumber \\
&\sum_{\alpha=\max \lbrace 0,\alpha'+\beta'-q-Q\rbrace}^{\min \lbrace l+L,\alpha'+\beta'\rbrace}\sum_{\varphi=0}^{o-l+O-L}\sum_{\varepsilon=0}^{m-o-q+M-O-Q}\sum_{\varphi'=\max \lbrace 0,\varphi+\varepsilon-m+o+k-k'-M+O+K-K'+q+Q+l+L-l'-L'\rbrace}^{\min \lbrace \varphi+\varepsilon,o+k-k'-l'+O+K-K'-L'\rbrace}G_{\rm{X}}\nonumber \\
&\sum_{\tau=0}^{k}\sum_{\nu=\max \lbrace 0,1-\tau \rbrace}^{n-k}\sum_{\chi=0}^{K}\sum_{\omega=0}^{N-\max \lbrace K,1+\chi \rbrace}\sum_{\tau'=\max \lbrace 0,\tau+\nu-n+k' \rbrace}^{\min \lbrace k', \tau+\nu\rbrace}\sum_{\chi'=\max \lbrace 0, \omega+\chi-N+K'\rbrace}^{\min \lbrace K', \omega+\chi\rbrace}\frac{f_{\rm{H}}(n,N,k,K,\tau,\nu,\chi,\omega)}{g(n,N,\tau+\nu,\omega+\chi,\tau+\nu,\omega+\chi)}
\nonumber \\
&f_{\rm{H}}(n,N,k',K',\tau',\tau+\nu-\tau',\chi',\omega+\chi-\chi')\,d(n,N,\tau+\nu,N-\omega-\chi,\eta_{\rm{det}}),
\end{align}
\end{widetext}
and
\begin{widetext}
\begin{align}\label{pncFinalX}
&p_{\rm{nc}}^{\rm{X}}=\sum_{n=0}^{\infty} \sum_{m=0}^{\infty}\sum_{N=0}^{\infty} \sum_{M=0}^{\infty}\,p_n\,q_m\,p_N\,q_M \sum_{k=0}^{n}\sum_{x=0}^{k}\sum_{y=0}^{n-k}\sum_{o=0}^{m}\sum_{k'=\max \lbrace x,o+k-m \rbrace}^{\min \lbrace n-y,o+k \rbrace}\sum_{K=0}^{N}\sum_{X=0}^{K}\sum_{Y=0}^{N-K}\sum_{O=0}^{M}\sum_{K'=\max \lbrace X,O+K-M \rbrace}^{\min \lbrace N-Y,O+K \rbrace}\nonumber \\
&\sum_{l=0}^{o}\sum_{q=0}^{m-o}\sum_{L=0}^{O}\sum_{Q=0}^{M-O}\sum_{l'=0}^{o+k-k'}\sum_{q'=0}^{m-o-k+k'}\sum_{L'=\max \lbrace 0,q+Q+l+L-q'-l'-M+O+K-K' \rbrace}^{\min \lbrace O+K-K',q+Q+l+L-q'-l'\rbrace}\sum_{\alpha'=0}^{l'+L'}\sum_{\beta'=\max \lbrace 0,1-\alpha' \rbrace}^{q+Q+l+L-l'-L'-\max \lbrace 0,1+\alpha'-l'-L'\rbrace}\nonumber \\
&\sum_{\alpha=\max \lbrace 0,\alpha'+\beta'-q-Q\rbrace}^{\min \lbrace l+L,\alpha'+\beta'\rbrace}\sum_{\varphi=0}^{o-l+O-L}\sum_{\varepsilon=0}^{m-o-q+M-O-Q}\sum_{\varphi'=\max \lbrace 0,\varphi+\varepsilon-m+o+k-k'-M+O+K-K'+q+Q+l+L-l'-L'\rbrace}^{\min \lbrace \varphi+\varepsilon,o+k-k'-l'+O+K-K'-L'\rbrace}G_{\rm{X}}\nonumber \\
&\sum_{\tau=0}^{k}\sum_{\nu=\max \lbrace 0,1-\tau \rbrace}^{n-k}\sum_{\chi=0}^{K}\sum_{\omega=\max \lbrace 0,1-\chi \rbrace}^{N-K}\sum_{\tau'=\max \lbrace 0,\tau+\nu-n+k' \rbrace}^{\min \lbrace k', \tau+\nu\rbrace}\sum_{\chi'=\max \lbrace 0, \omega+\chi-N+K'\rbrace}^{\min \lbrace K', \omega+\chi\rbrace}\frac{f_{\rm{H}}(n,N,k,K,\tau,\nu,\chi,\omega)}{g(n,N,\tau+\nu,\omega+\chi,\tau+\nu,\omega+\chi)}
\nonumber \\
&f_{\rm{H}}(n,N,k',K',\tau',\tau+\nu-\tau',\chi',\omega+\chi-\chi')\,d(n,N,\tau+\nu,\omega+\chi,\eta_{\rm{det}}).
\end{align}
\end{widetext}
Note that the differences between $p_{\rm{c}}^{\rm{X}}$ and $p_{\rm{nc}}^{\rm{X}}$ are in the limits of the summation index $\omega$ and in the arguments of the function $d$. Moreover, we note that we have found strong numerical evidence that $p_{\rm{c}}^{\rm{X}}=p_{\rm{c}}^{\rm{Z}}$ and $p_{\rm{nc}}^{\rm{X}}=p_{\rm{nc}}^{\rm{Z}}$ hold, but we have not been able to show it analytically by comparing Eq.~\eqref{pcFinal} and Eq.~\eqref{pncFinal} to Eq.~\eqref{pcFinalX} and Eq.~\eqref{pncFinalX}.

Now, with Eq.~\eqref{probabilityQNDZ}, Eq.~\eqref{pcFinal}, Eq.~\eqref{pncFinal}, Eq.~\eqref{pcFinalX} and Eq.~\eqref{pncFinalX} we have all the quantities that are required in order to evaluate the secret key rate of the protocol, which is given by Eq.~\eqref{AdequateSkrAppendix}. 

\subsection{Derivation of the secret key rate in the case of $\eta_{\rm{det}}=1$ and $\tau=0$}\label{app:Unit}
The unit efficiency secret key rate formula can be obtained by noting that in the quantities $p_{\rm{QND}}$, $p_{\rm{c}}^{\rm{Z}}$, $p_{\rm{nc}}^{\rm{Z}}$, $p_{\rm{c}}^{\rm{X}}$ and $p_{\rm{nc}}^{\rm{X}}$, given by Eq.~\eqref{probabilityQNDZ}, Eq.~\eqref{pcFinal}, Eq.~\eqref{pncFinal}, Eq.~\eqref{pcFinalX} and Eq.~\eqref{pncFinalX} only those terms in which the power of $(1-\eta_{\rm{det}})$ and $(1-\eta'_{\rm{det}})$ is $0$ contribute to the sums. This gives restriction on the indices appearing in the power of the terms $(1-\eta_{\rm{det}})$ and $(1-\eta'_{\rm{det}})$. Systematically examining the different cases one by one, we can rule out most of possibilities for the indices. For this, we need to keep in mind that if a sum happens to have a smaller number in the upper limit than that in the lower limit, then the corresponding term gives no contribution to the sum. 

Using the recipe given above and the formulas Eq.~\eqref{probabilityQNDZ}, Eq.~\eqref{pcFinal}, Eq.~\eqref{pncFinal}, Eq.~\eqref{pcFinalX} and Eq.~\eqref{pncFinalX}, we have that
\begin{equation}\label{unitPQnd}
p^{[\eta_{\rm{det}}=1,\,\tau=0]}_{\rm{QND}}=\frac{p_1\,(2\,q_2\,(1-\eta_{\rm{ch}})+3\,q_1\,\eta_{\rm{ch}})}{48},
\end{equation}
\begin{equation}\label{unitPc}
p^{\rm{Z}[\eta_{\rm{det}}=1,\,\tau=0]}_{\rm{c}}=p^{\rm{X}[\eta_{\rm{det}}=1,\,\tau=0]}_{\rm{c}}=\frac{p_1^2\,q_1^2\,\eta^2_{ch}}{1024},
\end{equation}
and
\begin{equation}\label{unitPnc}
p^{\rm{Z}[\eta_{\rm{det}}=1,\,\tau=0]}_{\rm{nc}}=p^{\rm{X}[\eta_{\rm{det}}=1,\,\tau=0]}_{\rm{nc}}=0.
\end{equation}
It is interesting that $p_{\rm{nc}}^{\rm{Z}}$ turns out to be $0$ when we set $\eta_{\rm{det}}=1$ and $\tau=0\,\rm{s}$. In the case considered in the calculations above, Alice and Bob both measured $H$ (horizontal) polarizations in their $Z$ measurement and the two QND measurements succeeded on both side. This means that Charlie's BSM will get input signals in $V$ (vertical) polarizations from both sides. An error can only occur when they apply the bit flip (see step 5 of the protocol), which only happens when Charlie detects a singlet state. But due to the Hong-Ou-Mandel effect the photons will always go to the same arm, therefore in this ideal case Charlie cannot obtain a singlet detection. The same argument holds for the $X$-basis as well. Plugging Eq.~\eqref{unitPQnd}, Eq.~\eqref{unitPc} and Eq.~\eqref{unitPnc} into Eq.~\eqref{AdequateSkrAppendix} we obtain the secret key rate formula for the unit detection efficiency case, given by Eq.~\eqref{skrFormulaUnitEff}.

\bibliographystyle{apsrev}

\end{document}